\documentclass[12pt, reqno]{amsart}

\usepackage{amsmath}
\usepackage{amssymb} 
\usepackage{array}
\usepackage{enumitem}
\usepackage{color}
\definecolor{battleshipgrey}{rgb}{0.52, 0.52, 0.51} 
\usepackage{hyperref}
\hypersetup{
    bookmarks=true,         
    unicode=false,          
    pdftoolbar=true,        
    pdfmenubar=true,        
    pdffitwindow=false,     
    pdfstartview={FitH},    
    pdftitle={Essay on axioms of quantum mechanics},    
    pdfauthor={Wolfgang Bertram},     
    pdfsubject={Subject},   
    pdfcreator={Creator},   
    pdfproducer={Producer}, 
    pdfkeywords={keyword1} {key2} {key3}, 
    pdfnewwindow=true,      
    colorlinks=true,       
    linkcolor=battleshipgrey,          
    citecolor=battleshipgrey,        
    filecolor=battleshipgrey,      
    urlcolor=battleshipgrey          
}
\usepackage{verbatim} 

\usepackage{pstricks-add}
\usepackage{pst-func}
\usepackage[utf8]{inputenc}

\theoremstyle{plain}
\newtheorem{theorem}{Theorem}[section]
\newtheorem{lemma}[theorem]{Lemma}
\newtheorem{proposition}[theorem]{Proposition}

\newtheorem{definition}[theorem]{Definition}
\theoremstyle{remark}
\newtheorem{remark}{Remark}[section]
\newtheorem{example}{Example}[section]

\numberwithin{equation}{section}

\newcommand{\bA}{\mathbb{A}}

\newcommand{\bB}{\mathbb{B}}

\newcommand{\K}{\mathbb{K}}
\newcommand{\bK}{\mathbb{K}}

\newcommand{\bR}{\mathbb{R}}
\newcommand{\R}{\mathbb{R}}
\newcommand{\bC}{\mathbb{C}}

\newcommand{\Z}{\mathbb{Z}}
\newcommand{\N}{\mathbb{N}}
\newcommand{\PP}{\mathbb{P}}

\newcommand{\cX}{\mathcal{X}}

\newcommand{\cG}{{\mathcal G}}

\newcommand{\cS}{\mathcal{S}}

\newcommand{\cH}{\mathcal{H}}

\newcommand{\cR}{\mathcal{R}}

\newcommand{\Gl}{\mathrm{Gl}}
\newcommand{\GL}{\mathrm{GL}}
\newcommand{\SL}{\mathrm{SL}}
\newcommand{\UU}{\mathrm{U}}

\newcommand{\Hom}{\mathrm{Hom}}
\newcommand{\Herm}{\mathrm{Herm}}

\newcommand{\Aut}{\mathrm{Aut}}

\newcommand{\End}{\mathrm{End}}

\newcommand{\Gras}{\mathrm{Gras}}

\newcommand{\Graph}{\mathrm{Graph}}
\newcommand{\Lag}{\mathrm{Lag}}

\newcommand{\id}{\mathrm{id}}

\newcommand{\Aherm}{\mathrm{SHerm}}

\newcommand{\CR}{\mathrm{CR}}
\newcommand{\tr}{\mathrm{trace}}

\newcommand{\af}{\mathrm{af}}

\newcommand{\Der}{{\mathrm{Der}}}

\newcommand{\setof}[2]{\{ #1 \mid #2\}}

\newcommand{\inv}{^{-1}}
\newcommand{\msk}{\medskip}
\newcommand{\ssk}{\smallskip}
\newcommand{\nin}{\noindent}

\begin{document}

\title[An essay on the completion of quantum theory.I]{An essay on the completion of quantum theory.
\\ I: General setting}



\author{Wolfgang Bertram}

\address{Institut \'{E}lie Cartan de Lorraine \\
Universit\'{e} de Lorraine at Nancy, CNRS, INRIA \\
B.P. 70239 \\
F-54506 Vand\oe{}uvre-l\`{e}s-Nancy Cedex, France\\
{\tt \url{http://www.iecl.univ-lorraine.fr/~Wolfgang.Bertram/}} }

\email{\tt{wolfgang.bertram@univ-lorraine.fr}}

\begin{abstract}
We propose a geometric setting of the axiomatic mathematical formalism of
quantum theory. 
Guided by the idea that 
understanding the mathematical structures of these axioms is of similar importance as was historically
the process of understanding the axioms of geometry,
we  complete the spaces of observables and of states in a similar way as in classical geometry
linear or affine spaces are completed by projective spaces. In this sense, our theory can be considered
as a ``completion of usual linear quantum theory'', such that the usual theory appears as 
the special case where a reference frame is fixed once and for all. 
In the present first part, this general setting is explained. Dynamics (time evolution) will be discussed in subsequent work.
\end{abstract}

\maketitle

\vskip 10mm

\begin{center}
{\sl
Dedicated to the memory of  \href{https://de.wikipedia.org/wiki/Tobias_Brandes}{Tobias Brandes} (1966 -- 2017) }
\end{center}

\bigskip

\section*{Preamble}

Since our years of study in G\"ottingen, Tobias and I had a plan to write, some day, a book,
``our'' book on quantum mechanics. Our paths separated after the Diplom: Tobias became a physicist, and I, 
a mathematician. We believed that we had time to carry out our project. But we had not.  

\msk
What I'm going to write up here, is my version of what might have been  a draft for the first chapter of this book. 
There are a lot of excellent textbooks on quantum mechanics, and our aim cannot  be, and never was,
 to copy them, or to cook a new one by mixing ingredients taken from them. Rather, by writing the book we
 would have wished to find answers to our own questions  -- the present version is certainly
a biased choice of questions, 
the one of a mathematician, and Tobias is no longer here to correct and complete it by a physicist's 
view. I'm well aware that the text is tentative, piecemeal, and possibly may appear altogether beside the point. 
My only excuse is that, from a purely mathematical point of view, the ideas exposed in the following seem
natural, and are kind of unavoidable.
I cannot  claim that Tobias would have signed this text, but I'm  sure that in the universe where he is now,
he will forgive me 
for quoting his name in relation with ideas and speculations that, certainly, are not quite standard in our universe.

\section{Introduction}

\subsection{Quantum Mechanics: axioms {\sl versus} interpretations}
Whereas the {\em interpretation} of Quantum Mechanics is a hot topic -- there are at least 15 different
\href{https://en.wikipedia.org/wiki/Interpretations_of_quantum_mechanics}{mainstream interpretations}\footnote{hyperlinks are in grey in the electronic version of this text}, an unknown
number of 
\href{https://en.wikipedia.org/wiki/Minority_interpretations_of_quantum_mechanics}{other interpretations},
and thousands of pages of discussion --, it seems
that the {\em mathematical axioms} of Quantum Mechanics are much less controversial: the
\href{https://en.wikipedia.org/wiki/Dirac%E2%80%93von_Neumann_axioms}{Dirac-von Neumann axioms}
are generally accepted to be their definite version
(\cite{D, vN}). 
Although I find exiting and interesting the discussion on ``interpretations'', I do not feel qualified to contribute to it.
As a mathematician I feel more competent to comment on the axiomatic and formal structure of quantum mechanics:
without being irrespectful towards Dirac and von Neumann, I find surprising that the ``definite'' form of the axioms has been
fixed  85 years ago,  shortly after the main discoveries of quantum theory had been made,
and that since then essentially nothing has been changed.
 The whole discussion seems to turn around the ``interpretation'' of a  theory whose formal mathematical structure
 is defined once and for all,
 without taking seriously into consideration that the axiomatic foundations may be questionable. 
This calls for comparison with the history of  
\href{https://en.wikipedia.org/wiki/Foundations_of_geometry}{\em axioms of geometry}:
Euclide's axiomatic construction
 of geometry is certainly among the greatest achievements of the human mind in ancient history;
however, sticking to the axioms too closely prevented men for a long time
 from discovering non-Euclidean geometry. The rapid development of
modern mathematics was possible only after mathematicians had questioned the structure of Euclide's axioms.
Could   something similar occur with the axioms of quantum theory? 
I think this possibility cannot be completely excluded. 

\ssk
Of course, I neither claim that the Dirac-von Neumann axioms were ``wrong'' (they can be no more ``wrong'' than Euclide's),
nor  to have a full-fledged
counter-proposition, like Hilbert had when proposing his
``\href{https://en.wikipedia.org/wiki/Hilbert%27s_axioms}{Grundlagen der Geometrie}'',
putting Euclide's axioms onto a rigorous and modern base. 
More modestly, I just want to point out 
 that such possibilities may indeed exist, by presenting some tentative framework;
it is then a matter of discussion between physicists and mathematicians to judge whether this deserves to be
investigated further, 
and if so,  to
improve it and  leading by iteration to a
kind of optimal version, hopefully in less time then it took to progress from Euclide's to Hilbert's vision of geometry. 

\ssk
In a nutshell, my proposition is to ``complete quantum theory'': since its present form is a {\em linear} theory, it calls for
completion by some non-linear space, just like Euclidean geometry calls for completion by {\em projective spaces}. 
This proposition is presented in Section \ref{sec:main}. 
Before presenting it, some more preliminary remarks.

\subsection{The universe of mathematics, and the mathematical universe}
Tobias was not the first and not the last to put forward the idea that ``physics {\em is} mathematics'' (I remember him
exposing this idea
 to me on a paper napkin in G\"ottingen): Roland Omn\`es discussed such kind of idea in his
book ``Converging Realities'' \cite{O2}, saying:
{\em I suggest the name ``physism'' for the philosophical proposal that considers the foundations of mathematics as
belonging to the laws of nature.}
More recently, this idea has been advanced by
Max Tegmark  (\cite{Teg}), who calls it the
\href{https://en.wikipedia.org/wiki/Mathematical_universe_hypothesis}{\em Mathematical Universe Hypothesis (MUH)},
that is: {\em Our external physical reality is a mathematical structure}, and (loc.\ cit., p. 357):
{\em The MUH implies that mathematical existence equals physical existence.}
As a mathematician, I feel quite  happy with this, and I like to take it as a heuristic principle, that is, as a
 welcome source of inspiration.
The MUH suggests that physicists and mathematicians approach the same thing from different sides:
physicists may call it the  ``Mathematical Universe'', and mathematicians may call it the ``Universe of Mathematics''.
Seen from the mathematician's side,
 the axioms of quantum theory are part of the universe of mathematics, and finding their ``optimal'' form is not
so much a matter of expedience, 
but rather an intrinsic mathematical question, whose
importance is comparable to the one of the foundations of geometry.  
Indeed, my feeling is that these two questions are much more deeply related to each other than visible at present.

\subsection{Form and content} 
Mathematicians tend to focus on the {\em formal} structure of the universe, 
on structures and relations, whatever their
``meaning'' or ``content'' may be.  As Hilbert put it once, referring to his
``Grundlagen der Geometrie'' (\cite{R}, p. 57):
{\em ``One must be able to say at all times -- instead of points, straight lines, and planes -- tables, chairs, and beer mugs.''}
Von Neumann (following the G\"ottingen spirit)  defined the fundamental notions of quantum theory,
{\em state} and {\em observable}, in a purely formal way as {\em rays in a Hilbert space} (``table''), respectively as
{\em self-adjoint operator in a Hilbert space} (``beer mug''). Thus,
passing  from the  ``classical''  to the  ``quantum world''  is often presented by the following
schema:

\msk
\begin{center}
\begin{tabular}{|l|c|c |}
  \hline
   & classical  & quantum \\
  \hline
 state  & point (element of a set)  & ray in Hilbert space  \\
 \hline
 observable  & (real) function on the point set & self-adjoint operator  \\
  \hline
\end{tabular}
\end{center}
\msk

\nin
Another version of this schema, in terms of 
\href{https://en.wikipedia.org/wiki/C*-algebra#C.2A-algebras_and_quantum_field_theory}{$C^*$-algebras}, reads as follows:

\ssk
\begin{center}
\begin{tabular}{|l|c|c |}
  \hline
   & classical  & quantum \\
  \hline
 state  & point  &  normed positive functional on a $C^*$-algebra   \\
 \hline
 observable  & (real) function  & Hermitian element of a $C^*$-algebra  \\
  \hline
\end{tabular}
\end{center}
\msk

\nin
This pattern is clean and neat,  and there exist many excellent textbooks unfolding it in
detail, both from the point of view of mathematics and of physics. As already said above, it is not our aim 
to reproduce them.

\subsection{Plan} 
The pattern presented above looks clean and neat,
but it is unsatisfying if you want to understand the ``structure of the mathematical
universe''  -- 
 fundamental notions are defined via a {\em construction} (``take a Hilbert space or a $C^*$-algebra, 
  and do this and that...''),
and not via intrinsic properties and relations. 
In Section \ref{sec:class2q}, we develop this critizism in more detail, and then 
present ingredients that might permit to formulate other axiomatics (main sections: 
\ref{sec:class-revisted} and \ref{sec:main}), essentially equivalent to
the Dirac-von Neumann axioms, but 
opening a window towards possible
new developments, by  indicating what structure could be omitted or altered when wishing to start a trip into
``non Dirac-von Neumannian quantum mechanics''. 
The present text deals with the ``general language'' of quantum theory, whose main vocabulary is
``state'' and ``observable''. In the subsequent second part, I will try to include {\em dynamics} into this
theory (unitary time evolution). 

\subsection{Geometry of quantum theory -- history}\label{sec:history}
Before starting the mathematical discussion, let me very briefly sketch the history of our topic, from my own
(admittedly subjective) viewpoint. 
Reading letters and texts by von Neumann (\cite{Re, V}), I have the impression that nobody shared the dissatisfaction 
with his schema more than he himself.  One the one hand, 
 together with Jordan and Wigner, he investigated the possibility of constructing
quantum mechanics by using only the ``algebra of self-adjoint operators''  -- which is not an {\em associative}
algebra, but (as we say nowadays) a {\em Jordan algebra}, with the symmetrized product
$a \bullet b := \frac{ab+ba}{2}$. I have been interested myself in the mathematical theory of
 Jordan algebras for a long time, and much of what
follows is motivated by this research. 

\msk
On the other hand,  von Neumann writes in a letter to Garrett Birkhoff, in 1935 (\cite{Re}, p. 59):
{\em I would like to make a confession which may seem immoral: I do not believe absolutely in Hilbert space any more.}
He then attacks, together with Birkhoff, his deep and beautiful work on the lattice theoretic approach, completed later
by contributions of
other outstanding mathematicians, and
presented  in lectures by George Mackey giving rise to the monograph \cite{V}. This monumental work
 is a major step in understanding mathematical structures
underlying quantum mechanics, and it answers in many respects the critisizm that I shall formulate below
  ({\it cf.}\ in particular the long notes to Chapter
IV in \cite{V}, and \cite{L}).  

\msk
However, reading \cite{V}, one ends up with the impression that the effect of this huge work is only to justify exactly the
Dirac-von Neumann axioms as given before: we gain the satisfaction that they can be deduced from more
general and more abstract principles. 
But nothing more --
there seems to be no ``window'' that could be opened, comparably to opening Euclidean geometry towards
non-Euclidean ones. 
Possibly, this feeling guided another generation of theoretical physicists, 
Aerts and his school on the one hand (\cite{A99}, where the term ``completed quantum mechanics'' is used in a sense
different from ours, and \cite{A09}), and on the other, Kibble, 
followed by Ashtekar and Shilling, and by
Cirelli, Gatti, and Mani\`a, and others (cf.\ references in \cite{Be08a, Be08b}),
 who instead of lattice theory used (infinite dimensional) differential geometry to investigate the
geometry of the ``state manifold'', the projective Hilbert space $\PP (H)$. 
This so-called ``delinearization program'' has
also influenced my own approach \cite{Be08a, Be08b}, on which the present text is based.
As far as I see, all of these authors pleading for a ``geometric approach'' to quantum theory 
have common aims and motivations, clearly formulated in \cite{CGM}:
{\em ``The delinearization program, by itself, is not related in our opinion to attemps to construct a non-linear extension of QM with operators that act non-linearly on the Hilbert space H. The true aim of the delinearization program is to free the mathematical foundations of QM from any reference to linear structure and to linear operators. It appears very gratifying to be aware of how naturally geometric concepts describe the more relevant aspects of ordinary QM, suggesting that the geometric approach could be very useful also in solving open problems in Quantum Theories.''
}

%

 \section{From classical to quantum}\label{sec:class2q}

Without going too much into details, here is what I would like to say as a mathematician, or as a ``geometrician'', 
about the basic pattern presented above.

\subsection{The ``classical side''}
Classical geometry deals with sets, say $M$, carrying additional structure having a ``geometric flavor''
(such as: manifold, symplectic or Lorentzian structure, and so on). 
The most elementary actors are the {\em points} of $M$, $p\in M$, which we call also {\em pure states}.
Note that there is no ``distinguished'' point in $M$, no ``origin''. 
However, one may object that ``points'' often appear to be a fiction, since they have no extension at all; it would be more
realistic to replace points by {\em probability measures $\mu$}, also called {\em mixed states}, 
on $M$. Then it would be a matter of
convenience to describe the correct topological, or measure-theoretic properties that one likes to impose.
Anyhow, points $p$ may be identified with the corresponding point-mass, or Dirac measure, $\delta_p$, and finite
convex combinations of Dirac measures represent mixed states coming from a finite number of pure states.

\ssk
An {\em observable} is a real-valued function $f:M \to \R$ (in presence of additional structure, usually 
 assumed to be continuous, or measurable, or smooth, and so on).
Denote by $F(M)$ or $F(M,\R)$ your space of observables (say, for the moment, the space of all real valued functions); then this space
carries a rich structure: it is a {\em vector space}, by pointwise addition and multiplication by scalars, and a
{\em commutative algebra}, by pointwise multiplication of functions, and there is a {\em partial order}: we may speak
of {\em positive functions}.
Note that all of these structures simply come
from the corresponding ones of real numbers $\R$, since everything is defined pointwise. 
You just loose two things: $\R$ is a {\em field}, but $F(M,\R)$ is not (it's just a {\em (commutative) ring}), 
and the order on $\R$ is {\em total}, but the one on
$F(M,\R)$ is not (it's only  {\em partial}). 

\ssk
Next, states and observables naturally are {\em in duality} with each other: an observable $f$ can be 
{\em evaluated} at a point $p$, just by taking the value $f(p)$.
If we work with mixed states (measures $\mu$), the same holds: if you take the view of defining a measure $\mu$
as a certain linear form on $F(M,\R)$, then the value is denoted by $\mu(f)$; if you use classical
measure theory, you will rather write  $\int_M f d\mu$, but in the end this amounts to the same.
In this context, $f$ may be called a {\em random variable}, and the value $\mu(f)$ is its {\em expectation value}.
When $\mu$ is a Dirac measure $\delta_p$, then this value is always ``sharp'' (there is no variance), but in general we
have to use the language of {\em probability theory}, as usual, e.g.,  in classical statistical mechanics.
On a conceptual level, already at this point a serious problem becomes visible: 
the ``problem of infinities'' -- 
certain measures attribute to certain
functions the value $\infty$ (which is {\em not} a real number), or no value at all.

\ssk
This is the basic set-up; much more could be said, and according to what you focus on, your theory will take
different shapes. For instance, noticing that evaluation at a {\em pure} state is an {\em algebra morphism}
$F(M,\R) \to \R$, you will be interested in {\em kernels} of the point evaluations, which are certain {\em ideals} of the algebra;
pursuing this (and replacing $\R$ by $\mathbb C$ or other fields), 
you are lead towards formalisms used in algebraic geometry. 
On the other hand, keeping to  real numbers, and
noticing that measures are {\em positive linear forms} on $F(M)$, you are lead to look 
at the vector space $S(M)$ of {\em signed measures}, which is a subspace of the dual vector space $F(M)^*$,
and to realize that the Dirac measures are {\em extremal points of the convex cone of positive
functionals}. This leads to  duality of topological vector spaces, order structures, and to functional and convex analysis. 
Both viewpoints are extremely important in modern mathematics.

\subsection{The ``quantum side'', and the ``superposition principle''}
Concerning the ``quantum side'', 
Varadarajan opens his book \cite{V}  by the phrase:
{\em As laid down by Dirac in his great classic \cite{D}, the principle of superposition of states is the fundamental concept on which  quantum theory is to be erected. }
It is not easy to find a clear explanation of what this principle means -- Dirac himself writes (in \cite{D} p.\ 15):
{\em The superposition process is a kind of additive process and implies that
states can in some way be added.}
Transposed to the classical picture drawn above, this would mean that we could  ``add'' two pure states (points), and
the result would be another pure state (point): that is, the manifold $M$ would be something like a linear space, with ``addition'' map
assigning to a pair of points a third one. 
Thus, the passage from classical to quantum would  resemble a procedure imposing some additional structure on
the pure state space, turning it into something similar to a linear (=vector) space. 
This is not too far from a valid formal definition --
today, we say much shorter: {\em a pure state is a ray in a (complex) Hilbert space $H$}, 
so that the {\em set of pure states} is nothing but
the {\em projective space} $\PP(H)$ associated to $H$.
Indeed, elements of a projective space cannot simply be ``added'', but projective spaces do bear certain relations 
with linear spaces, and ``superposition'' refers to reminiscence of this kind of linearity in quantum theory.
For instance, two different points in a projective space define a unique {\em projective line} joining them, which is 
the set of superpositions of these two points (but the parametrisation of this line is not unique). 
Summing up, the state manifold $M$ becomes, on the quantum side, not quite a flat, linear space, but something 
related, a (complex) projective space $\PP(H)$. 
{\em Projective geometry} thus becomes part of quantum theory. 
This observation has triggered the geometric approaches to quantum mechanics  mentioned above (subsection \ref{sec:history}).

\ssk
This apparently clear geometric picture suffers a setback when we wish to extend it to {\em mixed states}: as
on the classical side, one can speak of ``mixed states'', again defined  as formal convex combinations of pure
states. However, one now must take care not to confuse  such a formal convex combination with the superposition defined by
the same coefficients! 
 It is not clear what kind of ``geometric object'' the set of these general states then is: it is not a projective space, but still
one would expect it to remember somehow the ``superposition principle'', that is, to
 be some kind of geometry sharing properties with projective geometries -- some kind of
``generalized projective geometry''.  Indeed, here we are lead to intrinsically mathematical questions concerning the structure of
the Universe of Mathematics -- and related to my own research.\footnote{ I have called, in \cite{Be02},
 ``generalized projective geometries''
the precursors of the ``Jordan geometries'' from \cite{Be14}. The approach is quite different from the lattice theoretic one
developed in \cite{V}, cf.\  subsections \ref{sec:history} and \ref{ssec:vN}.}

\ssk
Back to the quantum side, let's now discuss the  {\em observables}: 
in the basic scheme, observables are represented by {\em self-adjoint operators on the Hilbert space $H$} (in general, 
\href{https://en.wikipedia.org/wiki/Unbounded_operator}{unbounded operators}
 --  but let us, for the moment, prescind from this).
States and observables are related with each other by a kind of {\em duality}, which in contrast to the classical case is now
of {\em quantum probabilistic nature}:
instead of a sharp ``value of the observable $A$ in the pure state $\psi$'', we just can speak of its
{\em expectation value}, which is the number given by the formula (where $\langle u,v\rangle$ is the scalar product in $H$)
\begin{equation}\label{eqn:exp1}
\langle A \rangle_\psi = \frac{ \langle \psi, A \psi \rangle}{\langle \psi, \psi  \rangle } \, ,
\end{equation}
or, more generally, of the {\em probability distribution} of the values, including the 
second, $\langle A^2 \rangle_\psi  - \langle A \rangle_\psi ^2$, and
higher 
\href{https://en.wikipedia.org/wiki/Moment_%28mathematics%29}{moments}.
Although the expression $\langle A \rangle_\psi$ looks more complicated than the classical $f(p)$, it still is additive
in $A$, so that observables, just as in the classical case, form a vector space, with the usual operator sum being
the same as ``pointwise sum''. However, the formula is not ``multiplicative'' (i.e., not compatible with the composition of operators).
The sum of operators thus seems to be the clear analog of the sum of functions from the classical case, whereas a clear 
interpretation of the product gets lost. 
The formula for $\langle A \rangle_\psi$ is apparently not ``linear'' in the variable $\psi$; all the more it is remarkable that
the operator $A$ itself acts linearly on $\psi$ -- the (complex)
 linearity of $A$  is the surprising feature of quantum theory, and indeed it
is the mathematical core of  the ``principle of superposition''.  
Whereas on the classical side there is {\em just one} source of linearity, on the quantum side
there seem to be {\em two} such sources, one on the level of
observables, the other on the level of states,  which somehow appear to be
 compatible with each other.
The precise formulation 
of this compatibility condition is  rather subtle --  there
are at least two ways to formulate it, corresponding to the two ways of presenting the classical-quantum scheme given above,
but in either way there is no such thing as ``superposition of mixed states'' (only of pure ones).

\ssk
The first way is by identifying a mixed state $W$, formal convex combination of orthonormal 
pure states $\psi_i$ weighted by scalars $w_i\in
[0,1]$ such that $\sum_i w_i =1$, with the corresponding ``density matrix'', the operator represented by the diagonal matrix
given by the $w_i$ with respect to the $\psi_i$.
Then the expectation value of $A$ in the mixed state $W$ is given by
\begin{equation}\label{eqn:exp2}
\langle A \rangle_W = \mbox{trace}(WA) .
\end{equation}
This formula is {\em linear in $W$}, and even {\em bilinear in $(A,W)$}. 
However, because of the normalizations, the density matrices do not form a linear space, but just a convex set, so the
term ``bilinear''  has to be taken with some care.

\ssk
The second way of interpreting these things, also going back to von Neumann, is to forget the Hilbert space $H$ and
to express everything in terms of the algebra $\bA$ of (say, bounded) operators on $H$, and in a next step 
taking for $\bA$ more general types of associative algebras.
Technically, one usually requires that $\bA$ be a 
\href{https://en.wikipedia.org/wiki/C*-algebra}{\em $C^*$-algebra}. Most importantly, this means that $\bA$
carries an {\em involution}, the {\em adjoint map} $a \mapsto a^*$, and the observables then
are the fixed points of this involution ($a^* = a$,  self-adjoint elements).
The notion of state now becomes a derived notion: a state is a
{\em  normalized positive linear functional} $\mu : \bA \to {\mathbb C}$, that is, a $\mathbb C$-linear map such that
$\mu(a^* a)$ is real and non-negative, for all $a \in \bA$, and $\mu({\bf 1})=1$, where $\bf 1$ is the unit element
of $\bA$. 
The states form a convex cone, and the pure states are the extremal elements of this cone. 
In this picture, the analogy with the classical picture appears very clearly: 
the commutative algebra $F(M)$ is replaced by the (complex) non-commutative algebra $\bA$, 
states are in both cases certain positive linear functionals, so there is a natural duality with the observables
(interpreted in the classical case as exact value, and in the quantum case as expectation value, which is not sharp
in general, even if the state is pure).
Summing up, the following table  refines the schema given in the preceding section:
\begin{center}
\begin{tabular}{|l|c|| c |c |}
  \hline
   & classical  & $\begin{matrix} \mbox{quantum:}\\  \mbox{Hilbert space } H \end{matrix}$ & 
   $\begin{matrix} \mbox{quantum:}\\  C^*\mbox{-algebra } \bA \end{matrix}$
    \\
  \hline
 pure state  & point $p\in M$  & ray $[\psi]$  in $H$ &  extremal  \\
 (mixed) state & probability measure $\mu$ & density matrix $W$ & 
$ \begin{matrix} \mu:\bA \to {\mathbb C}\\ 
 \mbox{normed, positive} \end{matrix}$   \\
 \hline 
 observable  & function $f:M\to \R$ & self-adjoint operator $A$ & $a\in \bA$ with $a^*=a$ \\
  \hline
 duality & $f(p)$, $\mu(f)$ & $\langle A \rangle_\psi$,  $\langle A \rangle_W$ 
   & $\mu(a)$ \\
 \hline 
\end{tabular}
\end{center}

\msk
\nin
This axiomatic framework proved to be powerful and robust since over 80 years. 
Why put it in question?

\subsection{Questions}
Stepping back and looking at the axioms from the point of view of the ``universe of mathematics'', 
reasons of dissatisfaction may be:

\begin{enumerate}
\item
Key notions (state, observable) are not introduced as
primitive objects, but are
  defined by a construction  using something else (Hilbert space $H$,
$C^*$-algebra $ \bA$).  In other words,
basic objects used in these constructions (vectors $\psi \in H$, elements $a \in \bA$)
appear to be auxiliary: they
 do not have a physical interpretation 
(only the rays $[\psi]$, resp.\ the elements with $a^*=a$, do).
\item
The axioms are ``ungeometric'': this seems unavoidable when we define objects {\em by their 
construction}, und not by properties and relations. For instance,
 the linear structure of $H$, resp.\ the bilinear product of $\bA$ are imposed by
decree, and the  ``superposition principle''  comes out of this construction in a fairly indirect way.
\item
Imposing by decree a linear structure also implies postulating the existence of an {\em origin}
($A=0$, resp.\ $a=0$), 
 and of a {\em unit}
($A=I$, the identity operator, resp.\  $a={\bf 1}$, the unit of the algebra). 
These very distinguished ``observables''  do not look ``physical'', but rather seem to reflect 
some kind of ``convention''. What is their true status?
\end{enumerate}

\nin
These three items are interwoven with each other. In the following, I shall use Item (3) as
 ``line of attack''. 

\subsubsection{On the classical side}
In the classical schema,
clearly the two observables  $f=0$ and $f=1$ (constant functions) play a very special r\^ole, and one may doubt if
they deserve to be called ``observable'': the function $f=0$ {\em defines} the ``origin'' and the function $f=1$ {\em defines} the
``unit'' with respect to which all other functions,  and hence all ``measurements'', are expressed. 
On the physics side, this raises deep questions about choices of unit systems,
about existence of absolute zeros, and whether these values and 
choices are constant all the time, whether they are
 ``canonical'' or peculiar to our particular place in the multiverse, and so on.
 Since these questions touch foundational issues of physics, I think that 
on the maths side, too, we should take them seriously. 
Thus, the least to say is  that these two functions don't look like
 observables of the same sort as the others: their status seems to be different. 
 
\ssk
Related to this item, a mathematician may regret that the duality between observables and states is somewhat unperfect:
there are constant functions, but no ``constant states''; there are pure states, but the ``pure observables'' (functions taking value
$1$ at one point and $0$ else) play no r\^ole. Of course, there are {\em analytic} reasons for this: the ``pure observables'' are not
continuous and would be of measure zero -- at least, in all ``continuous models''. 
Likewise, the ``yes-no questions'' (indicator functions $1_A$) are not continuous, hence out of scope of continuous
models: are they ``observables'', or ``states''? 
Taking continuous models  to be the only game in town 
excludes  from the outset to encompass discrete models (see below, \ref{ssec:contvsdiscrete}).

 \subsubsection{On the quantum side}
The ``scaling problem'', or ``problem of the linear structure'', gets more involved on the quantum side than
in the classical case, because of the double origin of its linear structure, see above.
The ``observables'' $0$ and $1$ (the identity operator) play a very distinguished role also on the
quantum side: again, they seem not to be ``observables like the others''.
In \href{https://en.wikipedia.org/wiki/Quantum_logic}{\em quantum logic}, they represent truth values
``always false'' and ``always true'' -- which clearly is a rather particular status, hardly comparable with
observables like position or momentum.

\subsection{Continuous versus discrete}\label{ssec:contvsdiscrete}
Although this is not the main topic of the present work, let me say some words on this item, questioning the classical
pattern. 
The question whether the universe should be seen as a ``continuum'', or as a ``granular 
({\em discrete}) structure''   is fundamental for choosing our mathematical model. The discussion already lasts  over
2000 years, and sometimes one opinion prevailed, and sometimes the other: 
the hypothesis of a ``granular'' structure is attributed to Democritus;
in ``classical physics'', the universe appeared to be continuous;  nowadays, in the ``quantum era'', it  appears to 
be discrete again, and maybe tomorrow, quantum-continuous...
For a mathematician, the lesson to be
drawn is that we should be ready to offer good models for {\em both} issues. 
But in practice we only have good models for the {\em continuum model}, and not for the discrete case (with its
``worst case'': the one of a {\em finite} set).
Indeed, in the continuous case, we have the whole  of {\em classical differential calculus} at our disposition,
which, combined with sophisticated functional analytic methods gives formidable strength to the approach of analyzing
a ``space'' $M$ via function algebras and measure theory on $M$. Its
power comes from {\em duality}, here: duality between function spaces and spaces of measures or distributions.
In the {\em finite} case, the distinction between functions and measures becomes more or less a fiction: both
  $F(M)$ and $S(M)$ then are the same as $\R[M]$, the vector space with base indexed by elements of $M$, which carries a
 canonical scalar product 
$\langle f,g \rangle = \sum_{p\in M} f_p  g_p$, so the same gadget may be considered as ``function'', or as
``signed measure'', whatever you prefer. 
This gives a purely algebraic model, which of course reflects nothing of the kind of properties
of ``infinitesimal calculus''.
In the general discrete case, similar remarks hold.
Thus, to get a sufficiently rich theory, we should ask:
{\em is there some  way to implement infinitesimal calculus in the discrete case?}
I  consider this to be a very interesting question -- already in the realm of purely ``classical'' mathematics! 
Giving a serious answer would take  too much place here.
I have been doing research in this domain for a certain time -- see \cite{Be17a} for an overview;
let me just say here that positive answers to this question do exist, and I believe they are relevant for the
structure of the Universe of Mathematics.


\section{The idea of ``completion''}\label{sec:idea}

We explain the basic idea of how to complete a linear theory.
In this chapter, we complete the space of
classical observables, and in the next, the one of quantum observables and states.

\subsection{Declaration of basic principles}
In my own research, I  regularly return to the ``problem of choice of units and origins'' (item (3) mentioned above).
Any ``geometric'' theory starting by attributing  a very special
r\^ole to  one, or two, or $n$, points, is flawed right from the beginning. 
I believe that somewhere in the constitution of the Universe of Mathematics it is written:  {\em All points are created equal}.
This axiom I hold for self-evident. It has corollaries, 
and to secure and formulate them it may be necessary to elect governments 
deriving their just powers from the consent of the governed. 
For instance,
students learn in linear algebra courses that one may need to fix bases in order to define matrices; but we should be able to change bases
(just like governements), or even to
get entirely rid of them.
Next, sometimes one wishes to get rid also of the origin of a vector space: this gives an
 {\em affine space}, which is nothing but a ``vector space with forgotten origin''.\footnote{ the idea is simple, but teaching
it to students is not --   the interested reader may
look \href{http://www.iecl.univ-lorraine.fr/~Wolfgang.Bertram/WB-affinespaces.pdf}{here}
 for some remarks on this...}
Likewise, in more abstract situations one may wish to get rid of ``neutral elements'' or ``zero points'' or ``unit elements'': 
this is exactly the approach I advocate also for the ``geometry of quantum mechanics''.
For symmetry reasons, in quantum mechanics this strategy must be applied both to observables and states, 
featuring the duality between them. In this respect, duality in quantum mechanics appears to be more perfect than
duality in the classical setting: the dual parts are of the same nature, they appear to be ``self-dual''.

\ssk
This proposition could be qualified ``conservative'':  no 
 fancy new gadgets are introduced, but  rather we renovate classical, if not old-fashioned, furnitures like affine spaces.
More specifically, we will follow an old route pointed by the observation that the geometry of an affine space inevitably
calls to be ``completed''  by adding ``points at infinity'', the so-called ``horizon'', to obtain
some kind of more symmetric space, the {\em projective space}. My proposition is to place
axiomatic quantum theory in the framework of a geometric space that ``completes'' an associative ($C^*$-) algebra,
or a Jordan algebra,
exactly like a projective space completes a usual vector of affine space.  
A good deal of my mathematical work has been devoted to such questions (see \cite{Be00, Be02, BeKi, Be08a, Be08b,
Be14, Be17b}), and to summarize, I can say that there is no mathematical obstruction to achieve this: such
geometries do exist, and moreover, nothing is lost by the procedure of ``completing''.
 What is gained? - 
On the maths side, a more homogeneous
 and more symmetric picture, allowing to look behind the horizon; on the physics side -- 
 I don't know; but by comparison with classical geometry, and by the philosophy
of the MUH, one may speculate that
the gain could  be non-zero.  Future may tell.

\subsection{The classical side revisted}\label{sec:class-revisted}
Let's start again by
looking at the classical pattern, and by investigating more closely the special r\^ole of elements
such as $1,0$ and $\infty$.
First,
the element $1 \in \R$ defines the ``canonical'' basis in $\R$ with respect to its ``canonical'' origin $0 \in \R$.
Let's forget the basis and look at $\R$ just as an abstract one-dimensional vector space: there is a distinguished
$0$, but no distinguished $1$.
Then the dual space is also a one dimensional vector space, but both spaces should be distinguished from each other: to remind this,
let us write $'\R$ for the ``original'' space, and $\R'$ for its ``dual'' space, and continue to write $\R$ for our old friend. 
When $v\in \, '\R$ and $\phi \in \R'$, let us write $\langle v,\phi \rangle = \phi(v)$ for the value taken by $\phi$ on $v$.
This defines a  {\em pairing}
$\R' \times \, '\R \to \R$, $(\phi,v)\mapsto \langle v, \phi \rangle$.  Expressions such as
$\langle v, \phi \rangle w$ are defined, and
since we know that this is just another way to write the product $v\phi w$, we also know how to deal with brackets in such
itereated products.\footnote{ We discover, then, that the {\em inverse} $r^{-1}$ of a non-zero element
$r \in \, '\R$ belongs to $\R'$: it is the unique element such that $\langle r,r^{-1}\rangle = 1$. }
Summing up, forgetting the element $1$ emphasizes the r\^ole of {\em duality}: it
forces us to distinguish $'\R$ and $\R'$.
On the level of functions, 
we now have so speak of ``original functions'' $f:M\to \,  '\R$, and of ``dual functions'' $\phi:M\to \R'$, giving rise by
paring to a function $\langle \phi, f \rangle : M \to \R$.
For the moment, let us think of both kinds of functions as  ``observables''.
However, since obviously
$F(M, \, '\R)$ and $F(M,\R')$ are sort of dual to each other (indeed,
injecting a bit of language introduced in Appendix \ref{app:AP}, 
the pair $(F(M, \, '\R),F(M,\R'))$  is an archetypical
example of an {\em associative pair} and of a {\em Jordan pair}),
we will have to ask ourselves if it wouldn't be more appropriate to identify one of the two spaces rather with some kind of
``space of states''. 

\ssk
In a second step, let us forget both elements $0$ and $1$: this means to consider $\R$ just as an
{\em affine line}, that is, the {\em one-dimensional affine space}, $\R^\af$. By picking up any two distinct
points $a,b \in \R^\af$, we can identify $\R^\af$ with $\R$ such that $a$ corresponds to $0$ and $b$ to $1$:
namely, $r \in \R$ corresponds to the ``barycenter''
\begin{equation}\label{eqn:aa}
c = (1-r)a+rb \in \R^\af .
\end{equation}
Conversely, given a triple $(a,b,c)$ of points in $\R^\af$, we recover $r$ as {\em  (division) ratio}
\begin{equation}\label{eqn:ab}
r = \frac{c-a}{b-a} =: R(c,b,a) \in \R ,
\end{equation}
and $c=R(c,1,0)$.
The ratio is an \href{https://en.wikipedia.org/wiki/Invariant_(mathematics)}{\em invariant} of affine geometry.
In the same way,
our space of observables $F(M,\R)$ is turned into an affine space $F(M,\R^\af)$ by forgetting the functions 
$0$ and $1$.  Relations (\ref{eqn:aa}) and (\ref{eqn:ab}) remain valid, pointwise:
 any two functions $f_0,f_1:M \to \R^\af$,  taking different values at each point,
can take the roles of ``origin'' and ``unit''. Instead of $f \in F(M,\R)$, we  now consider the 
 triple $F= (f,f_1,f_0)$ as ``observable''. The number describing the observable $F$
in the pure state $p$ is
the ratio  $R(f(p),f_1(p),f_0(p))$, that is,  we  define  the ``value of $F$ at $p$'' by
\begin{equation}
F(p)  :=  R(f(p),f_1(p),f_0(p)) =
\frac{f(p) - f_0(p)}{f_1(p)-f_0(p)}.
\end{equation}
This formulation ensures that, even if values and choices of units and origins throughout the multiverse are 
uncommitted, the {\em mathematical form of laws} has a common description (as long 
as ratios are accepted as physical meaningful -- which is possibly the oldest idea of exact science).\footnote{
One may object that $f=0$ is distinguished by being a {\em constant} function, whereas $f_0$ will in general not be
constant.  But my point is precisely that this distinction rather reflects a convention, and not a ``fact of nature'':
to take account of this, any choice of $f_0$ also defines a modified {\em action of the group of bijections}, or
of diffeomorphisms or whatever, on functions, such that $f_0$ becomes invariant, i.e., ``constant'', under this action.
More generally, the notion of ``constant section'' of a vector bundle is not absolute, but depends on additional
structure, such as, e.g., affine connections.}
 
\ssk
In the universe of mathematics, these first two steps force us to make a third one: the duality aspect from the first step
has been lost in the second, and only by going to the {\em projective line} we can harvest the benefits of both steps
together -- indeed the
\href{https://en.wikipedia.org/wiki/Duality_%28projective_geometry%29}{\em duality principle of projective geometry}
is one of the highlights of classical geometry, and it makes projective geometry clearly superiour to affine geometry. 
So, instead at $F(M,\R^\af)$, let us look at the space $F(M,\R \PP^1)$ of functions $f:M\to \R \PP^1$ with values
in the  \href{https://en.wikipedia.org/wiki/Real_projective_line}{\em real projective line}  $\R \PP^1$.
For the present purposes, it will be sufficient to define $\R \PP^1$ 
simply as the ``one-point compactification'' 
$\R \cup \{ \infty \}$ of $\R$, by adding a single ``point at infinity''
 (topologically, $\R \PP^1$ is a circle, but for the moment we are not interested in topology).
We will use  two basic facts about projective geometry (see Appendix \ref{app:P} for some mathematical explanations): 
\begin{enumerate}
\item
removing an arbitrary hyperplane $H$ from a projective space $X$, an affine space $X \setminus H$ remains
(cf.\ Theorem \ref{th:P1}); in our case:
removing an arbitrary point $a$ from $\R \PP^1$, an affine line over $\R$ remains, denote it by
$U_a = \R \PP^1 \setminus \{ a \}$;
\item
picking up {\em two} different points $a,b$ in the projective line, 
$b$ may serve as origin in the affine space $U_a$, and $b$ may serve as origin in  $U_a$:
thus we have two (one-dimensional) vector spaces $(U_a,U_b)$. Now, {\em these two vector spaces 
are dual to each other} (see subsection \ref{ssec:P3}).
In other words, for any such choice of $(a,b)$, the pair $(U_a,U_b)$ is a ``model'' for $(' \R, \R')$.
\end{enumerate} 
Both (1) and (2) can be used to associate to a quadruple $(a,b,c,d)$ of elements  of $\R \PP^1$ a scalar in $\R$:
the ratio of $(a,b,c)$ in $U_d$, and the duality $\langle a,b\rangle$ in $(U_c,U_d)$.
It is remarkable, then, that both procedures give the {\em same} number, namely the 
famous 
\href{https://en.wikipedia.org/wiki/Cross-ratio}{\em cross-ratio} of the four values (see Appendix \ref{ssec:CR}): 
\begin{equation}\label{eqn:CR}
\CR(a,b ; c,d) =
\frac{(c-a)(d-b)}{(c-b)(d-a)} = 
 \frac{c-a}{c-b} : \frac{d-a}{d-b} = \frac{R(a,b,c)}{R(a,b,d) } \, .
\end{equation}
(To memorize notation: of the six possible differences, only $a-b$ and $c-d$ do {\em not} appear; the semicolon
reminds this.) 
The cross-ratio is a rich and subtle projective invariant. It contains all information given by
 preceding constructions since
\begin{equation}\label{eqn:CR2}
R(a,b,c) = \CR(a,b;c,\infty) ,\quad
\frac{a}{b} = \CR(a,b;0,\infty), \quad
a = \CR(a,1;0,\infty) .
\end{equation}
Again, the construction can be carried out pointwise:   by (1),
a completely arbitrary function $h$, or $f_\infty:M \to \R \PP^1$ can serve as ``horizon function'', or ``infinity function'': the set
of all functions never taking the same values as $h$ forms an affine space, another copy of our $F(M,\R^\af)$.
When $h$ is the function $f(x)=\infty$ for all $x \in M$, then we get back our old ``standard realization''; but now we
have also infinitely many other choices. 
Picking up {\em three} arbitrary functions, denoted by $f_0,f_1,f_\infty:M\to \R\PP^1$,
subject to the condition that at any point they take pairwise different values, we  use
them as ``reference triple'':
$f_1(p)$ as unit and $f_0(p)$ as origin in the affine space $U_{f_\infty(p)}$;  thus given any function
$f:M\to \R\PP^1$, we may define its ``value at $p$'' as the well-defined real number with respect to this reference 
triple, given  by the pointwise cross-ratio of the quadruple
$F = (f,f_1,f_0,f_\infty)$: 
\begin{equation}\label{eqn:CR2}
\langle F,p \rangle := \CR \bigl(f(p),f_1(p);f_0(p),f_\infty(p) \bigr) =
\frac{ f(p) - f_0(p)}{f_1(p) -f_0(p)} : \frac{f_\infty(p) - f(p)}{f_\infty(p) - f_1(p)} .
\end{equation}
But, as said above, the same formula also realizes item (2)!
It represents the bilinear pairing $F(M,\, '\R) \times F(M,\R')\to F(M,\R)$ when $(f_0,f_\infty)$ is fixed, and
in this context rather should be read 
\begin{equation}\label{eqn:CR3}
(f,g) \mapsto  \CR \bigl(f(p),g(p);f_0(p),f_\infty(p) \bigr) =
\frac{ f(p) - f_0(p)}{g(p) -f_0(p)} : \frac{f_\infty(p) - f(p)}{f_\infty(p) - g(p)} 
\end{equation}
Thus the pair $(g,f_\infty)$ represents an object ``dual'' to $(f,f_0)$, where the duality is given by the functional
``pointwise cross-ratio''. This suggests a shift in the understanding of the notion of ``state'': 
a state rather is a pair of dual functions, and an observable a pair of ``original'' functions, and this suggests
 to write
the whole gadget $F = (f,g;f_0,f_\infty)$ (call it ``obstate'')  rather  as a matrix
\begin{equation}\label{eq:matrix1}
F = \begin{pmatrix}  f & f_0\\ g & g_\infty \end{pmatrix} \, .
\end{equation}
Its  first row represents the ``observable aspect'', and the second row the ``state aspect'';  the second column
represents the ``reference system aspect'', and the first column its ``objective aspect''.
The cross-ratio
is invariant under exchange of rows, or of columns, and exchanging $f$ and $g$
(or $f_0$ and $g_\infty$) yields the inverse value. 
This shift of understanding furnishes  a robust concept of ``duality'' and of ``self-duality'', and it allows to separate this
from the
thorny problem of extracing a {\em scalar valued} pairing (using traces, integrals, measures -- see Appendix \ref{app:TDI}).

\ssk
Summing up, classical mathematics, and classical mechanics and other classical theories, 
could equally well be described by
replacing real valued functions by quadruples of $\R\PP^1$-valued functions, and by working with cross-ratios  instead of
values of single functions. Of course, this looks heavy and unnessarily complicated. And indeed, so it is, as long as
origins, units and infinities are considered to be fixed once and for all.
In classical mathematics, this assumption may seem reasonable;
but even then it might be interesting to pursue this idea since it opens new views on certain fundamental issues. 
With this perspective in mind, we mention that a further
property of $\R$ generalizes rather nicely to $\R \PP^1$: 
 the {\em order relation} of $\R$ gives rise to a 
 \href{https://en.wikipedia.org/wiki/Cyclic_order}{\em cyclic order} on $\R\PP^1$. 
Namely,
\begin{itemize}
\item
on the linear space $\R$, the order is 
  given by a {\em unary relation}: $0<x$,
 \item  
  on the affine space $\R^\af$, it is given by a {\em binary relation}, $x<y$, as usual,
\item
 on the projective space $\R \PP^1$, it is given by a {\em ternary relation}: 
a triple $(a,b,c)$ is {\em cyclically ordered} if $a<b$ in the affine space $U_c$.\footnote{Put differently and more formally:
the group $\PP \Gl^+(\R^2)$ has two open orbits in $(\R \PP^1)^3$:  one of them is the set of cyclically ordered triples.}
We then write $b \in ]a,c[$, thus defining {\em intervals} on $\R \PP^1$.
\end{itemize} 
\nin
Again, for functions, things carry over pointwise: what we get is a {\em partial cyclic order on the space of functions from
$M$ to $\R \PP^1$}.  The set of positive functions is generalised by intervals of this partial cyclic order
 (see \cite{Be17b} for more on this).  
Cyclic order and cross-ratio are related with each other: $\CR(a,b;c,d)$ is negative iff
$c$ lies in $]a,b[$ and $d$ in $]b,a[$, or vice versa, i.e., if  
the pair $(c,d)$ ``separates'' $(a,b)$. Some geometers (e.g., Coxeter) chose this separation relation as belonging to the structures
appearing in axiomatic foundations of geometry.

\subsection{From real to complex}\label{ssec:R2C}
Since quantum mechanics requires {\em complex} Hilbert spaces, and {\em complex} $*$-algebras, we may in a first step
replace {\em real} functions from the classical picture by {\em complex} functions, $f:M\to \bC$.
One may agree that this is just a ``trick'', since in the end the observables shall be real-valued.
Everything said in the preceding section goes through (except the cyclic order, of course): 
it suffices to replace the real projective line by the
{\em complex projective line}, $\bC\PP^1 = \bC \cup \{ \infty \}$ (which now topologically is a $2$-sphere, the
\href{https://en.wikipedia.org/wiki/Riemann_sphere}{\em Riemann sphere}).  
The cross-ratio is defined in the same way, and it is invariant under complex conjugation.
It follows that an ``obstate'' 
$F = (f,g;f_0,f_\infty)$ is real if, and only if, its cross-ratio is real. 
Now, it's a classical fact
 that $\CR(a,b;c,d)$ is real if, and only if, the four points $a,b,c,d$ lie on a 
 \href{https://en.wikipedia.org/wiki/Generalised_circle}{\em generalised circle},
that is, either lie on a circle in $\bC$, or on a real affine line. 
Thus we have two possibilities to define ``real obstates'':
\begin{enumerate}
\item
a quadruple of $\R \PP^1$-valued functions, as in the preceding subsection,
\item
a quadruple of $\bC\PP^1$-valued functions such that, at every point $p \in M$, the four values
lie on a generalised circle. 
\end{enumerate}
Let's call an obstate ``real'' in the first sense, and ``real-like'' in the second one.

\subsection{Antipode mapping}\label{sec:antipode}
The cross-ratio is invariant under the full projective group: it is a projective invariant. 
On the other hand, the dynamics of quantum mechanics is governed by the {\em unitary group}, which is much
smaller than the projective group. Thus at some point quantum mechanics requires to plug in some
additional structure. 
For instance, we may fix a scalar product on $\R^2$, or on $\bC^2$, say the standard scalar product
$\langle x,y \rangle = \overline x_1 y_1 + \overline  x_2 y_2$, and consider the induced {\em polarity} 
on $\bC \PP^1$, that is,
the {\em orthocomplement map} (where $J$ is the matrix given by (\ref{eqn:J}))
\begin{equation}
\alpha: \bC \PP^1 \to \bC  \PP^1, \quad [z] =\Bigl[ \begin{pmatrix}z_1\\z_2\end{pmatrix}\Bigr] 
\mapsto [z^\perp]= [ J \overline z] = 
\Bigl[ \begin{pmatrix} - \overline z_2\\ \overline z_1\end{pmatrix}\Bigr] .
\end{equation}
In the usual chart of $\bC \PP^1$, this map is given by $z \mapsto - \overline z^{-1}$;
but if we identify $\bC\PP^1$ with the Riemann sphere $S^2$, then $\alpha$ is rather represented by the
{\em antipode map} sending a point of the sphere to its opposite, or antipode point. 
The projective maps commuting with $\alpha$ are exactly those coming from the projective unitary group $\PP\UU(2)$.
Thus we can reduce the projective invariant cross-ratio to a two-point invariant of $\PP\UU(2)$:
in formula (\ref{eqn:CRR}), let $a=\alpha(y)$, $b=\alpha(x)$, then
\begin{equation}\label{eqn:CRRR} 
P(x,y):=
\CR (x,y;\alpha(y),\alpha(x)) = 
 \frac{ \langle x,y \rangle \cdot  \langle y,x \rangle }{\langle x,x \rangle \cdot \langle y,y \rangle}  = \cos^2(\phi(x,y)) ,
\end{equation}
where $\phi$ measures the angle between the vectors $x,y \in \bC^2$. 
Of course, the same holds for $\bC$ replaced by $\R$.\footnote{
The formula for $P(x,y)$ defines a {\em transition probability}, in the sense
of \cite{L98}, p.80, or \cite{L}, p.\ 31. 
Using (\ref{eqn:CRR}), the same
procedure can be applied to any projective space.}
Again, applying everything pointwise, these definitions carry over to function spaces: we can define
$\alpha(f)$ and $P(f,g)$ for functions.

\ssk
The antipode mapping $\alpha$ on $\bC\PP^1$ is {\em antiholomorphic}, just like the complex conjugation
$\tau(z) = \overline z$ of $\bC\PP^1$, whose fixed point set is $\R\PP^1$.
Since $\alpha$ and $\tau$ commute, the composition $\beta :=\alpha \circ \tau$ is the map induced by the
matrix $J$, given in the usual chart by
$z \mapsto -  z^{-1}$, which is a {\em holomorphic} map of order $2$. 
It has precisely two fixed points:
$i$ and $-i$. When picturing $\R \PP^1$  as equator of the sphere $\bC\PP^1$, these two fixed points shall
be pictured as {\em north and south pole}, and the points $0$ and $\infty$ on the equator could be called
{\em east and west pole}, and $1$ and $-1$ {\em front and back pole}.
The ``usual chart'' is stereographic projection from the west pole onto the tangent plane of the sphere at the
east pole 
(Figure \ref{fig:RS}). 
The four transformations  $\{  \tau,\alpha,\beta, \id \}$ form an abelian group (a Klein four group) acting on
$\bC\PP^1$.

\begin{figure}[h]
\caption{The Riemann sphere $\bC \PP^1$ with six poles.}\label{fig:RS}
\newrgbcolor{eqeqeq}{0.8784313725490196 0.8784313725490196 0.8784313725490196}
\newrgbcolor{aqaqaq}{0.6274509803921569 0.6274509803921569 0.6274509803921569}
\psset{xunit=0.6cm,yunit=0.6cm,algebraic=true,dimen=middle,dotstyle=o,dotsize=5pt 0,linewidth=0.8pt,arrowsize=3pt 2,arrowinset=0.25}
\begin{pspicture*}(-9.68,-4.78)(13.66,4.94)
\psplotImp[linewidth=1.8pt](-11.0,-8.0)(14.0,7.0){-1.0+1.0*y^2+0.0625*x^2}
\psplotImp(-11.0,-8.0)(14.0,7.0){-1.0+0.0625*y^2+0.5*x^2}
\psplotImp(-11.0,-8.0)(14.0,7.0){-16.0+1.0*y^2+1.0*x^2}
\psplotImp(-11.0,-8.0)(14.0,7.0){-1.0+0.0625*y^2+0.08333333333333333*x^2}
\psplot[linecolor=lightgray]{-9.68}{13.66}{(--2.--0.4294487507354019*x)/1.024295039463181}
\psplot[linecolor=lightgray]{-9.68}{13.66}{(--3.930370166695838-0.42694233953616334*x)/-4.810276799130294}
\psplot[linecolor=lightgray]{-9.68}{13.66}{(-3.9303701666958424--1.4512373789993456*x)/-2.060903212636129}
\begin{scriptsize}
\psdots[dotsize=8pt 0,dotstyle=*,linecolor=darkgray](0.,-4.)
\rput[bl](-0.08,-4.76){\darkgray{$S$ (south)}}
\psdots[dotsize=8pt 0,dotstyle=*,linecolor=darkgray](0.,4.)
\rput[bl](-0.12,4.24){\darkgray{$N$ (north)}}
\psdots[dotsize=8pt 0,dotstyle=*,linecolor=darkgray](-1.374686793247079,-0.9390898592677543)
\rput[bl](-1.1,-0.52){\darkgray{$F$}}
\psdots[dotsize=8pt 0,dotstyle=*,linecolor=darkgray](1.3746867932470865,0.9390898592677545)
\rput[bl](1.46,1.26){\darkgray{$B$}}
\psdots[dotsize=8pt 0,dotstyle=*,linecolor=darkgray](-3.435590005883215,0.5121475197315905)
\rput[bl](-6.0,0.5){\darkgray{$O$ (zero)}}
\psdots[dotsize=8pt 0,dotstyle=*,linecolor=darkgray](3.435590005883215,-0.512147519731591)
\rput[bl](4.16,-0.8){\darkgray{$W$ (infinity)}}
\psdots[dotsize=2pt 0,dotstyle=*,linecolor=darkgray](-0.040457868926124955,1.9355999417622334)
\rput[bl](0.04,2.06){\darkgray{$B'$}}
\psdots[dotsize=2pt 0,dotstyle=*,linecolor=darkgray](-8.379988211529422,-1.5608544482033058)
\rput[bl](-8.3,-1.44){\darkgray{$F'$}}
\end{scriptsize}
\end{pspicture*}
\end{figure}

\section{Completion of Quantum Theory}\label{sec:main}

This is the main section: we are going to explain the general setting ``completing''  usual, linear quantum
theory.
By ``usual'' formulation we mean the one in terms of a $C^*$-algebra $\bA$ (but we will not use all
properties of a $C^*$-algebra, only those which define a {\em $P^*$-algebra}, see 
Appendix \ref{app:P*}).  
 For some mathematical constructions and definitions we shall refer to the appendices.
The algebra $\bA$, respectively, its real subspace $\Herm(\bA)$, are completed by the following ``geometric spaces''
\begin{align}
\cG & := \Gras(\bA^2) = \{ x \subset \bA^2 \mid \, x \mbox{ (right) submodule, } x \not=0, x\not= \bA^2 \} ,
\\
\underline \cS & := \Gras_\bA (\bA^2) = \{ x \in \Gras(\bA^2) \mid \, x \cong \bA \},
\\
\overline \cS  &:= \Gras^\bA(\bA^2) =\{ x \in \Gras(\bA^2) \mid \, \bA^2/ x \cong \bA \},
\\
\cS & := \bA \PP^1 := \Gras_\bA^\bA(\bA^2) := \overline \cS \cap \underline \cS,
\\
\cR &:= \{ x \in \cS \mid \, x^\perp = J x \} \mbox{ (=Lagrangian variety of } \omega(u,v)=\langle Ju,v \rangle {\rm ) },
\\
\cR' & := \{ x \in \cR \mid \, \bA^2 = x \oplus x^\perp \} = \{ x \in \cR \mid \, \bA^2 = x \oplus Jx \} 
\label{eqn:realuniverse}
\end{align}
where $J$ is given by (\ref{eqn:J}) and $\perp$ the orthogonal complement with respect to the usual ``scalar product''
$\langle u,v \rangle = \sum_i u_i^* v_i$
on $\bA^2$. 
In the classical case $\bA = \bC$,
the spaces $\cG,\underline \cS,\overline \cS$ and $\cS$ all coincide with the Riemann sphere $S^2$, and the spaces
$\cR$ and $\cR'$ both coincide with the unit circle (equator) $S^1$. 
When $\bA$ is infinite dimensional, the inclusions 
$$
\cR' \subset\cR \subset \cS \subset \underline \cS \subset \cG
$$
are in general all strict, and one may consider them as inclusions of ``universes'', 
sitting inside each other like Matryoshka dolls. 
  The ``base points'' $0 = [(1,0)]$ and $\infty=[(0,1)]$ belong to all of them, and so does
the ``affine part'' $\{ [(1,a)]\mid a \in \Herm(\bA)\} \cong \Herm(\bA)$, which represents the ``algebra of (bounded)
observables'' from ``usual'' quantum mechanics, so that the nested sequence arises by adding more and more 
``points at infinity'' to the ``usual'' space. 
Having fixed the pair $(0,\infty)$, the ``natural chart'' $\bA \subset \cS$ generalizes stereographic
projection; but the ``set at infinity'' (the part of $\cS$ not covered by $\bA$) now is in general quite a big
set: it contains a distinguished point $\infty$, but also many others. If $\bA$ is finite-dimensional, then
$\bA$ will always be dense in $\cS$, but if $\bA$ is infinite-dimensional, then this need not be the case. 
We will describe two versions of the setting:
\begin{enumerate}
\item
a {\em weak}, or {\em projective setting} (subsection \ref{ssec:weak}): the structure is just given by $(\cS,\tau)$;
 its symmetry group is big (the whole projective group of $\cR$). This setting  suffices to define
 {\em expectation values} (first moments),
\item
a {\em  strong}, or {\em unitary setting} (subsection \ref{ssec:strong}): the structure is given by $(\cS,\tau,\alpha)$;
its symmetry group is smaller (essentially, unitary), and it permits to recast the whole of quantum theory,
including {\em higher moments}.
\end{enumerate}
As said in the introduction, this text is still preliminary and experimental: at present, it is not entirely clear to me
which parts of quantum theory really belong to the ``weak setting'', and which to the ``strong setting'', or maybe
to some intermediate setting yet to be defined. 

\subsection{The weak (projective) setting}  \label{ssec:weak}
This setting is given by the data $(\cS,\cR,\tau)$. Its symmetry group,
generalizing the usual projective group $\PP( \SL_2(\R))$,
 is described in the appendix, equation (\ref{eqn:autR}).
In this setting, it is appropriate to distinguish the projective line from its dual:

\subsubsection{Duality and self-duality}
The projective line over $\bA$, as well as the Hermitian projective line, are {\em self-dual}: they agree with their 
dual projective line, $\cS = \cS'$, $\cR = \cR'$, see Appendix \ref{app:P}.
However, both for mathematical and for philosophical reasons, we  shall separate,
whenever possible, two copies $(\cR,\cR')$, resp.\ $(\cS,\cS')$, considered to be ``dual to each 
other''.
In more technical terms, this means that
we  treat, whenever possible, the associative algebra $\bA$ as an {\em associative pair} $(\bA^+,\bA^-)=(\bA,\bA)$
(and the Jordan algebra $\Herm(\bA)$ as a {\em Jordan pair}); see Appendix \ref{app:AP}.
Still put differently, we try, as long as possible, to avoid using the unit element $1$ of $\bA$.

\subsubsection{Basic terminology: complete obstates}
We use the term ``obstate'' for ``observa\-ble-state'' to denote an entity incorporating ``observables'' and 
``states''. 

\begin{definition}
A {\em complete obstate}, or {\em obstate} for short, is a quadruple
$$
{\bf O}:=(A,W;A_0,W_\infty)
$$ 
such that:
$A,A_0 \in \cR$, and $W,W_\infty \in \cR'$, and 
$A_0 \top W_\infty$, $A_0 \top W$, $A \top W_\infty$ (where $\top$ means ``transversal'', see Appendix \ref{ssec:P2}).
The pair $(A,W)$ is called the {\em object part} of the obstate, or {\em objective obstate}, and the pair
$(A_0,W_\infty)$ is called the {\em reference part}, or {\em reference obstate}. 
The pair $(A;A_0)$ is called  {\em complete observable}, and the pair $(W;W_\infty)$ 
 {\em complete state}.
  This terminology is summarized by the 
``obstate matrix'':
\ssk
\begin{center}
\begin{tabular}{|l|c|c |}
  \hline
   & object part  & reference part \\
  \hline
 complete observable   &  $A$  & $A_0$  \\
 \hline
 complete state   &  $W$  & $W_\infty$  \\
  \hline
\end{tabular}
\end{center}
\end{definition}
 
\nin In ``usual'' quantum physics, the reference part is fixed once and for all, and then is  denoted just by
$(0,\infty)$. The idea of ``complete quantum physics'' should be that complete obstates
are the intrinsic objects to be studied. As long as the reference part is fixed, 
no deviation from usual quantum 
physics should appear, that is, we postulate the mathematical ``conservation rule'': {\em for a fixed reference part
$(A_0,W_\infty)$, the rules of complete quantum theory shall reduce to the rules of usual quantum theory}. 
(To ensure this, we have included the transversality assumption $A_0\top W_\infty$  in the definition.)
Since the unit element $\bf 1$ of the algebra does not appear in the reference part, 
this ``weak setting'' comprises all aspects of the usual theory that do not depend on, or do not require the choice of,
a unit element.

\ssk
We do not make any claims about ``interpretations'' or ``reality'' corresponding to a possible change of 
reference parts.
Indeed,  the reader may safely assume that $(A_0,W_\infty)$ is fixed once and for all.
 Maybe s-he will find, later on,
that it is much more convenient to assume that this is not the case, and that this is in much better keeping with some of the
current ``interpretations'' of usual quantum theory. This, possibly, could open the hypothetical window towards
``non Dirac-von Neumannian quantum theory'' -- which  should never be in contradiction with usual quantum physics.
(There might be apparent contradictions due to unclear terminology: before projective geometry was invented,
a phrase like ``two parallel lines intersect at infinity'' sounded contradictory.)

\subsubsection{Pure states; rank} 
Under our ``conservation rule'',  for a fixed reference system $(0,\infty)$, states $W$ shall correspond to density matrices, 
and thus pure states shall correspond to density matrices of rank one. 
The aim of the paper \cite{BeL} is to give a ``geometric interpretation'' of such concepts, in the general context of Jordan
geometries. Let me try to summarize the main ideas in the present, more special context: 
given two elements $x,y \in \cR$ that belong to some affine part $U_a$ of $\cR$, the locus of the real affine line
$[x,y]_a = \{ t x + (1-t) y \mid t \in \bR \}$ does in general heavily depend on the choice of $a$.
But for certain choices of the pair $(x,y)$, this locus does {\em not} depend on the choice of the affinization $a$: 

\begin{definition}\label{def:rank-one}
A pair $(x,y) \in \cR^2$ is said {\em of rank $1$}, or {\em of arithmetic distance $1$}, if $x \not=y$ and for all
$a,b \in \cR$ with $x,y \in U_a \cap U_b$:
$$
[x,y]_a \cap U_b = [x,y]_b \cap U_a .
$$
Then $\overline{[x,y]} := [x,y]_a \cup \{ \infty \}$ is a copy of the real projective line $\R\PP^1$ in $\cR$ that depends only
on $x$ and $y$, called an {\em intrinsic line in $\cR$}.
A complete state $(W;W_\infty)$ is said {\em pure} (and then we shall often write $(\psi;\psi_\infty)$, following a venerable tradition)
if the pair $(W,W_\infty)$ is of rank $1$.
\end{definition}

\nin
For instance, in a projective space $\R \PP^n$, every pair $(x,y)$ with $x\not=y$ is of rank 1:
{\em every} pair of distinct points defines a unique intrinsic projective line joining them. In sharp contrast, 
for higher rank geometries,
such as $\cR$, such lines can only follow very special directions (these directions lie on the extreme boundary of the 
``light cones'' that define the ``generalized conformal structure'' of $\cR$, see below, \ref{ssec:positivity}). 
Algebraically, saying that $(0,x)$ is of rank $1$ corresponds to saying that {\em $x$ is von Neumann regular}, 
or that we can find $\infty$ such that $x$ becomes an {\em idempotent} (see Def.\ \ref{def:idempotent}), or yet
that {\em $x$ generates a minimal  inner ideal}.
Likewise, in \cite{BeL} it is explained that higher rank is related to inner ideals that need not be minimal. 
Their geometric counterpart has been christianed {\em intrinsic subspace}.

\begin{remark}
The term {\em arithmetic distance} is due to 
\href{https://en.wikipedia.org/wiki/Hua_Luogeng}{L.-K.\ Hua}, who studied it for all series of finite-dimensional matrix geometries. 
There are interesting links between the arithmetic distance between $(x,y)$ and algebraic invariants of the torsors
$U_{xy}$. For instance, when $(x,y)$ is of rank $1$, then $U_{xy}$ is a solvable group with derived series having one non-trivial
term.
\end{remark}

\subsubsection{Expectation value of an obstate}
Assume $(A,W;A_0,W_\infty)$ is a complete obstate.
We have to extract a real number from these data, which for the fixed reference system $(A_0,W_\infty)$ shall coincide with the one given 
by  (\ref{eqn:exp1}) or (\ref{eqn:exp2}). 
Imperatively, this scalar has to  be given by a {\em scalar valued cross-ratio}:

\begin{definition}\label{def:EV}
The {\em expectation value} of the complete obstate $(A,W;A_0,W_\infty)$ is
\begin{align*}
\langle A,W; A_0,W_\infty \rangle :& = 
\tr (K_{A_0,W_\infty} (A,W ))   \\
& = \tr (\CR ( A,W; A_0,W_\infty )),
\end{align*}
where $K$ and the generalized cross-ratio $\CR$  are defined by eqn.\  (\ref{eqn:GCR1}) -- (\ref{eqn:GCR2}),
and $\tr$ denotes a trace on $\bA$ in the sense of Def.\ \ref{def:trace}. 
\end{definition}

\nin
This definition is {\em natural}, in the sense that it is invariant under the automorphism group 
$\Aut(\cS,\tau)$.  However, mind:
\begin{enumerate}
\item
{\em Traces} of linear operators always exist in finite dimension over a field, but 
not always in very general situations.
Indeed, this is not a ``quantum'' problem, but already appears in 
 the ``classical case'' (section  \ref{sec:class-revisted}):
associating a {\em scalar} to a pair (function, dual function) is some kind of {\em integration}, and already classical
integrals may lead to infinite values   (see Appendix \ref{app:TDI}).
\item 
For the formula from Definition \ref{def:EV} to reduce to (\ref{eqn:exp2}), in case $(A_0,W_\infty)=(0,\infty)$, we have to carefully 
distinguish a space from its dual space.
 If one misses that point, one would read the expression as
$\tr(AW^{-1})$.
\item
We cannot define ``second moments'' in the same way, since the  definition of an operator $A^2$ depends on the choice
of a unit element $1$, which is not given in the present setting. As far as I see, it is not possible to define
such higher moments in the present ``weak setting'': one needs more, and more rigid, structure to define them,
see below.
\end{enumerate}
\nin
The second item is related to the {\em normalization} which, in the usual theory,
 is necessary to write formula (\ref{eqn:exp2});
in our ``intrinsic'' formula in def.\ \ref{def:EV} {\em no} normalization is necessary 
(think of $W_\infty$ as the ``zero matrix'', which of course is not a density matrix itself, in the usual theory).

\ssk
Let us re-interprete this construction in a more geometric way for {\em pure} states $(\psi,\psi_\infty)$: in this case,
the intrinsic projective line ${\bf L} \cong \R\PP^1$ determined by the pure state contains already two distinguished elements,
$\psi$, and $\psi_\infty$. The observable $(A_0,A)$ defines two other distinguished points $(a_0,a)$ on $\bf L$:
namely, $a$ is the
unique point of $\bf L$ completing the affine line ${\bf L} \cap U_A$, and likewise for $a_0$. 
Now,
 the expectation value is the (classical)  cross-ratio of these four points on the line ${\bf L}  \cong \R\PP^1$:
\begin{equation}
\langle A,\psi; A_0,\psi_\infty \rangle  = 
\CR(a, \psi; a_0, \psi_\infty) .
\end{equation}
This is the analog of (\ref{eqn:exp1}). 
If $(A,A_0)$ 
happens to be already on $\bf L$ (so $a=A$, $a_0=A_0$),
 then the measurement is ``sharp'', but in general, this will not be the case,
and there will be higher moments (cf.\ below).
Note that for pure states we do not have to bother about problem (1)
mentioned above,  since  traces exist for rank-one operators.

\subsubsection{Axiomatic setting; superposition principle}
As said above,  in the usual setting,
Dirac's   ``superposition principle of quantum theory''
corresponds to assuming that observables are operators acting {\em linearly on a linear space}, or
that the ``(Jordan) algebra of observables'' carries a {\em bilinear} product. In our setting,
this property can be translated into the form of geometric axioms (cf.\ \cite{Be02, Be14}): it means that the geometry $(\cR,\cR')$ is an {\em affine pair geometry} -- 
every element $a \in \cR$ defines an {\em affine part $U_a$ of $\cR'$}, and vice versa, every $w \in \cR'$ defines an affine part
$U_w$ of $\cR$. In other words, the geometry is covered by ``affine charts'', which are part of its structure.
In an axiomatic ``geometrically complete quantum theory'', this property should be part of the axioms.
It then becomes a theorem (cf.\ \cite{BeL}) that the {\em intrinsic lines} form, 
in turn, another geometry, that is, they also have a local linear
structure, corresponding to the superposition principle. 
Thus a truly axiomatic presentation of ``completed quantum theory'' should be possible; but for pedagogical reasons it must
be postponed.

\subsubsection{Real versus complex}
Expectation values shall be real, and not complex. 
This can be achieved by a purely real theory, and
the setting presented so far does not (yet) really explain why {\em complex} numbers play such an important 
r\^ole in quantum theory, compared to the classical theory.
Indeed, everything said so far makes sense more generally when $\cR$ is the {\em Jordan geometry
corresponding to an abstract ordered Jordan algebra}, cf.\ \cite{Be17b} (except that the definition of the 
generalized cross-ratio becomes more involved if no associative structure is around). 
As far as I see, the true role of the complex numbers appears more clearly in the ``strong setting''. 
For the moment, we have the same
two options for formulating the ``complete'' theory as mentioned in subsection \ref{ssec:R2C}, and so far both of them appear to be reasonable:
\begin{enumerate}
\item
``real'':
we work in the universe of the Hermitian projective line $\cR = \cR'$; that is, all four components of 
$(A,A_0,W,W_\infty)$ shall  belong to $\cR$;
\item
``real-like'':
we  work in $\cS =  \bA \PP^1$, but 
we require that all four components of a complete obstate belong to a ``generalised circle'' (conjugate of the Hermitian projective line under the
projective group $\PP \Gl(2,\bA)$).
Since expectation values of quadruples are invariant under the projective group, this still ensures that
 all expectation values  are real.
 \end{enumerate}


\subsubsection{Positivity: cyclic order}\label{ssec:positivity}
In the $C^*$-algebra setting, it is part of the definition of states that they are {\em positive} linear functionals.
We have not included positivity in our definition of a complete state, since the precise formulation of such an
assumption is related to the question of ``interpretations'' of the formalism. First of all, in the projective setting,
the binary order relation generalizes to a {\em ternary}  relation.
As starting point, we use the binary partial order $\leq$ on $\Herm(\bA)$,
which exists by definition of a $P^*$-algebra (def.\ \ref{def:P*-algebra}),
 and then define a partial order
$\leq_c$ on each affine part $U_c$ (cf.\ \cite{Be17b}, Theorem 4.1), defining the ternary relation:

\begin{definition}
A complete state $(W;W_\infty)$ is called {\em positive with respect to a reference part $(A_0,W_\infty)$}, if the triple
$(A_0,W,W_\infty)$ is {\em cyclically ordered}, that is, if
$A_0 \leq W$ in the ordered vector space $U_{W_\infty}$.  
We say that $(A,A_0;W,W_\infty)$ is a {\em cyclically ordered obstate} if $(A,A_0,W_\infty)$ and
$(A_0,W,W_\infty)$ are cyclically ordered triples. This implies that the expectation value
$\langle A,W;A_0,W_\infty \rangle$ is positive. 
\end{definition}

\nin
As explained in \cite{Be17b}, the intervals on $\cR$ define a kind of {\em generalized conformal}, or {\em causal},
structure, modelled on the positive cone of $\Herm(\bA)$.


\subsection{The strong (unitary) setting}\label{ssec:strong} 
Now we add the following datum to the ``weak setting'':
the standard scalar product on $\bA^2$ defines an orthocomplementation map
$\alpha :\cS \to \cS$, $x\mapsto x^\perp$
 which is antiholomorphic and commutes with $\tau$, so that the holomorphic map
$\beta:= \alpha \circ \tau:
\cS \to \cS$ is again of order $2$. 
The data $(\cS,\tau,\alpha)$ define the ``strong setting''. 
There are no ``closed'' formulae for $\tau$ and $\alpha$, but as in the classical case,
the map $\beta$ is induced by the matrix $J = \bigl(\begin{smallmatrix}0&-1\\1&0\end{smallmatrix}\bigr)$, 
that is, $\beta = [J]$, 
so that in the usual chart, for $z \in \bA$,
\begin{equation}
\tau(z) = z^*, \qquad
\beta(z) = - z^{-1},\qquad
\alpha(z) = - (z^*)^{-1} .
\end{equation} 

\subsubsection{North and south pole.}
The map $J:\bA^2 \to \bA^2$ is diagonalizable over $\bA$: it has two eigenvectors $(i,1)$ and $(-i,1)$
with eigenvalues $i,-i$, so
\begin{equation}
\Bigr[ C^{-1} J C \Bigl] = \Bigr[ \begin{pmatrix} i &0\\0& - i \end{pmatrix}\Bigl] =
 \Bigr[ \begin{pmatrix} 1 &0\\0& - 1 \end{pmatrix}\Bigl], \quad
 \mbox{ where } C = \begin{pmatrix} i & - i \\ 1 & 1 \end{pmatrix} 
 \end{equation}
 (the matrix $C$ describes the {\em Cayley transform}, see below). 
Thus the map 
 $\beta= [J]:\cS \to \cS$ has precisely two fixed points, called {\em north pole} and {\em south pole},
\begin{equation}
N := [(i,1)], \qquad S:= [(-i,1)] . 
\end{equation}
Since one eigenvalue is the negative of the other,
the map $\beta$ acts by multiplication with $-1$ on the linear spaces $(\cS_N,S)$ and $(\cS_S,N)$, i.e.,
$\beta = (-1)_{N,S}$ (using notation (\ref{eqn:scalar})).
Thus $\beta$, and hence also $\alpha$, can be recovered from $(N,S)$, and we see that the data
$(\cS,\tau,\alpha)$ and $(\cS,\tau,N,S)$ are essentially equivalent.

\subsubsection{The canonical $S^1$-action.}
Since the data $(N,S)$ are canonical, not only the reflection map $(-1)_{N,S}$ is canonical, but every map of the
form  
 $\lambda_{N,S}$ with $\lambda \in S^1$. Indeed, these maps  
  commute  with
$\tau$ since $\tau \circ \lambda_{N,S} \circ \tau = \overline \lambda_{\tau(N),\tau(S)} =
(-\lambda)_{S,N} = \lambda_{N,S}$, hence preserve  $\cR$. Thus we get a canonical action
\begin{equation}
S^1 \times \cR \to \cR, \quad (\lambda,x) \mapsto \lambda_{N,S} (x) .
\end{equation}
In particular, $\beta$ has $i_{N,S}$ as a canonical square root: $i_{N,S}^2 = (-1)_{N,S}=\beta$. 

\subsubsection{Automorphism groups: unitary structure}
The strong setting is more rigid than the weak one, hence its automorphism group is smaller. 
Nevertheless, this group still has ``big orbits''. Let's explain this: 
The automorphism group of $\cS$ is $\PP \Gl(2,\bA)$; the one of the weak setting is
$\Aut(\cS,\tau) = \{ g \in \Aut(\cS) \mid g \circ \tau = \tau \circ g \}$ (cf. Appendix \ref{ssec:Hpl}),
and the one of the strong setting is
\begin{align*}
\UU:=\Aut(\cS,\tau,\alpha) & = \Aut(\cS,\tau) \cap \Aut(\cS,\alpha) \\
& = \Aut(\cS , \alpha) \cap \Aut(\cS,\beta) =
\PP\UU(2;\bA) \cap \Aut(\cS,\beta) \\
& = \PP
\Bigl\{ f= \begin{pmatrix} a & b \\ - b& -a \end{pmatrix} \mid \,
a^* a + b^* b = 1, a^*b - b^* a = 0 \Bigr\} 
\\
&= \PP
\Bigl\{ f= \begin{pmatrix} a & b \\ - b& -a \end{pmatrix} \mid \, a,b \in \bA, 
(a + i b) \in \UU(\bA)  \Bigr\},
\end{align*}
given by unitary operators $f :\bA^2 \to \bA^2$ such that $f J = J f$.
Via the Cayley transform,
this group is isomorphic to $\PP(\UU(\bA) \times \UU(\bA))$.

\subsubsection{The real unitary universe $\cR_{N,S}$, and strong obstates.}

\begin{definition}
We call {\em real unitary universe} 
the subset of $\cR$ given by all elements that are both transversal to $N$ and to $S$,
\begin{equation}
\cR_{N,S}:=
\cR \cap U_{N,S} = \{ x \in \cR \mid x \top N, x \top S \} .
\end{equation}
\end{definition}

\nin With $\cR'$ given by (\ref{eqn:realuniverse}), we always have $\cR_{N,S}\subset \cR' \subset \cR$. 
The real unitary universe is the space where strong obstates live:

\begin{definition}
A {\em strong obstate} is an obstate
$(A,W;A_0,W_\infty)$ such that
\begin{enumerate}
\item
$W_\infty$ and $A_0$ are antipodes of each other:
$W_\infty = \alpha (A_0)$,
\item
$A_0 \in \cR_{N,S}$ [by Theorem \ref{th:U-affine} below, all 4 components then belong to 
$\cR_{N,S}$].
\end{enumerate}
\end{definition}

\subsubsection{On the structure of the real unitary inverse}
Here are the most important results on the structure of $\cR_{N,S}$. 
They are special cases of more general and abstract results from \cite{BeKi2} (except for Theorem \ref{th:U-affine});
we will give more computational and down-to-earth proofs (using the Cayley transform) in \cite{Bexy}.
The first result, contained in \cite{BeKi2}, says that $\cR_{N,S}$ 
 ``is'' the unitary group.
This will be basic for our interpretation of unitary
time evolution \cite{Bexy}. 

\begin{theorem}
The real unitary universe $\cR_{N,S}$ carries a canonical {\em torsor structure}, that is, for any choice of origin 
$a \in \cR_{N,S}$ this set carries a group structure with unit element $a$ and
product $xz= x \cdot_a z$, such that with respect to any other
origin $y$, the product is given by $x \cdot_y z = xy^{-1}z$. Moreover, any of the groups thus obtained is isomorphic to
the unitary group $\UU(\bA)=\UU(\bA,*)$ (cf.\ Def.\  \ref{def:staralgebra}).  
\end{theorem}

\begin{theorem}\label{th:U-trans}
The automorphism group $\UU$ acts transitively on the real unitary universe $\cR_{N,S}$.
The stabilizer group of a point $o$ is isomorphic to $\UU(\bA)$, so that as homogeneous space,
$\cR_{N,S} = \UU.o  \cong (\UU(\bA) \times \UU(\bA))/\UU(\bA)$.
\end{theorem}

\nin Indeed,  any torsor acts transitively on itself by left or right translations, and these
always belong to the automorphism group (\cite{BeKi2}).

\begin{theorem}\label{th:U-affine}
The real unitary universe contains affine parts defined by all of its points:
for all $a \in \cR_{N,S}$, the set $\cR_a = \{ x \in \cR \mid x \top a \}$ is included in $\cR_{N,S}$.
\end{theorem}

\begin{proof}
In contrast to the preceding theorems (which are valid for any $*$-algebra), this result crucially relies on the 
``positivity'' property of a $C^*$-algebra (it is valid for the more general $P^*$-algebras, but not for general
$*$-algebras).
The essential point is that the affine formula for the Cayley transform can be defined on the whole affine part.
Again, full details will be given in \cite{Bexy}.
\end{proof}

\nin The space $\cR$ can be seen as  as an infinite dimensional manifold, see \cite{BeNe}, such that
the affine parts $\cR_A$ form open chart domains.
Thus the preceding theorem implies that $\cR_{N,S}$ is a union of open sets, hence  open in $\cR$.

\begin{theorem}\label{th:fd}
If $\bA$ is a finite dimensional $C^*$-algebra, then $\cR_{N,S}$ is a topological connected component of $\cR$.
In particular, if $\bA = M(n,n;\bC)$, then $\cR =\cR' = \cR_{N,S}$ is the unitary group $\UU(n) = \UU(n\bC)$. 
\end{theorem}

\begin{proof} If $\bA$ is finite dimensional, then the unitary group $\UU(\bA)$ is {\em compact}, hence 
$\cR_{N,S}$ is compact, hence closed.
Since $\UU(\bA)$ is also connected, it follows that $\cR_{N,S}$ is a connected component of $\cR$.
For the case $\bA = M(n,n;\bC)$, see also \cite{Be00}.
\end{proof}

\nin
If $\bA$ is infinite dimensional, the statement from the theorem will fail in general.
In group theoretic terms, this means that the $\Aut(\cR)$-orbit of a point $a \in \cR_{N,S}$ may be strictly bigger
than $\cR_{N,S}$, whereas in the finite dimensional case we have equality: the space $\cR_{N,S}$ then is what is
sometimes called a 
\href{https://en.wikipedia.org/wiki/Generalized_flag_variety#Symmetric_spaces}{\em symmetric $R$-space}.

\begin{theorem}\label{th:AB}
For each $a \in \cR_{N,S}$, the linear space
$$
\bA_a := \alpha(a)^\top = \{ x \in \cS \mid x \oplus \alpha(a) = \bA^2 \}
$$
carries the structure of an associative algebra, with zero vector $a$ and unit element
 $b:= i_{N,S}(a)$.
 This algebra  is
isomorphic to the asociative algebra $\bA$ with unit $ 1$.
Likewise, its real form $\bA_a^\tau$ is a Jordan algebra with unit element $b$ and zero vector $a$, isomorphic to
$\Herm(\bA)$.
\end{theorem}

\begin{proof}
Since the stabilizer group
$\UU(\bA)$ acts (via conjugation) by automorphisms on the algebra $\bA$, we may transport the algebra structure
from the algebra at the base point $0$ to any other point of $a$ of $\cR_{N,S}$, by transitivity.
The pair $(0,\infty)$ is mapped to $(a,\alpha(a))$,  the unit $1$ is then mapped to a point $b$.
The only thing which is not quite obvious is that then, necessarily, $b = i_{N,S}(a)$.
This, again, is proved using the Cayley transform, see \cite{Bexy}.
\end{proof} 

\begin{remark}
On may think of $\bA_a$, or rather of $\Herm(\bA_a)$ as a ``tangent algebra of the geometry at the point $a$'' (see
\cite{Be14}). 
The product in the algebra with neutral element $b$, and the unitary group law $\cdot_b$ with unit $b$, are
{\em dual to each other}, in the sense of {\em Cartan duality of symmetric spaces}: 
the Jordan cone at $b$ is kind of ``non-compact dual'' of the ``compact-like''
unitary group at $b$. This is related to the topic of Jordan-Lie algebras, see \ref{ssec:JL}. 
\end{remark}

The preceding results permit to reduce ``strong completed quantum theory'' to ``business as usual'': 
since all algebras $\bA_a$ are equivalent, we may (as long as $a$ is considered to be fixed), ``without loss of generaltity'',
assume that $a=0=[(0,1)]$ is the ``usual base point. For instance:

\subsubsection{Second and higher moments}
Given a strong obstate $(A,W;A_0,W_\infty)$, we may compute in the algebra
$\bA_a$ with $a=A_0$.
Let $AW$ be the product of $A$ and $W$ in this algebra, and $L_X:\bA_a \to \bA_a$, $Y \mapsto XY$ the operator of left multiplication
by    $X$.  Then the operator valued cross ratio and $L_{AW}$ coincide:
\begin{equation}
\CR(A,W;A_0,W_\infty) = L_{AW},
\end{equation}
and hence also their traces:
$\langle A,W;A_0,W_\infty\rangle = \tr (AW)$, so that expectation values are calculated in the algebra $\bA_a$ in the usual way.
Since $\bA_a$ is an algebra with binary product (and not just an associative pair), we now can form also all expressions of the
form $A^k$, $k\in \N$, and in particular we can define 
 as usual the second moment (variance) of the strong obstate, by
\begin{equation}
V(A,W;A_0,W_\infty) := \langle A^2 \rangle_W - \langle A \rangle_W^2 = 
\tr ( AWA) - (\tr(AW))^2  .
\end{equation}
If we assume that $\bA$ is a $C^*$-algebra, we can also define the probability distribution on $\R$
induced by the complete obstate, via the spectral theorem, in the usual way (cf.\ e.g.,  \cite{vN, L}).
In a similar way, all other properties and constructions can be carried over from $\bA$ to $\bA_a$.

\subsubsection{Conceptual approach: geometry of Jordan-Lie algebras}\label{ssec:JL}
Presenting things by simply transferring everything to ``business as usual'', as phrased above, 
 is not very conceptual, nor satisfying, but at least we see that a geometric,
base point-free setting for the geometry of quantum theory exists (and this is all I wanted to show at present). 
Possibly, a better understanding of what is going on here can only be achieved in connection with studying
{\em dynamics}: the unitary (Schr\"odinger) evolution on the one hand (Part II \cite{Bexy}), and, much more difficult, a mathematical
analysis of the {\em measurement process} from a geometric viewpoint (Part III ?). 
Mathematically, as far as I see,  the ``strong setting'' is the geometric counterpart of the algebraic structure of a {\em Jordan-Lie algebra}
(cf.\ Appendix \ref{app:P*}, and Part II \cite{Bexy} for a more detailed introduction). 
In \cite{E} and in \cite{L98}, Jordan-Lie algebras are taken as mathematical starting point to develop quantum
theory; thus on purely mathematical grounds, I  think it should be important to fully understand what the 
``geometry of a Jordan-Lie algebra'' really is. 
In particular, the interplay between the weak, projective, setting, and the strong, unitary, setting is rather subtle, and the explanations
given above are certainly insufficient. 

\subsubsection{Completed qubits}\label{ssec:qubit}
The smallest non-commutative real universe is the 
\href{https://en.wikipedia.org/wiki/Qubit}{\em qubit-space},
the completion of the $4$-dimensional Jordan algebra
 $\Herm(2,\bC)$ (cf.\ also the table in Appendix \ref{ssec:Hpl}). This Jordan algebra  is isomorphic
to Minkowski space $\R^{3,1}$, and its positive cone is the Lorentz cone. Its completion is precisely 
 the {\em conformal compactification of
Minkowski space}, often used in relativity theory. (Of course, this is just a pure coincidence, isn't it?)

\subsubsection{Towards the second chapter: dynamics}\label{sec:2chapter2}
At this point, the first chapter of our book would end.
Almost everything the reader is waiting for is still missing, so for sure, s-he would be impatient to start reading the
second chapter: so far, there is no Schr\"odinger equation (no dynamics, no time at all), no Heisenberg relation, not even
$\hbar$ did show up.  So, I hope to meet you soon again.

\section{Some concluding remarks}\label{sec:conclusion}

I don't know how many chapters the book may have, and if it will  ever be finished. 
From a mathematical point of view, I think the ideas presented here are kind of inevitable, and should be
pursued until they are fully understood.
Meanwhile, here are a few more mathematical remarks.

\subsection{Duality,  self-duality, and von Neumann}\label{ssec:vN}
For my taste, one of the most interesting aspects of the theory explained so far
 is the interplay between {\em duality}, and {\em self-duality}:
in order to understand and to organize projective geometry, or geometry of quantum theory, 
\href{http://www.iecl.univ-lorraine.fr/~Wolfgang.Bertram/Atiyah-Duality.pdf}{duality is a necessary principle}; 
but then it turns out that certain
structures are {\em self-dual}. The self-dual structures are an important part of the landscape.
If I'm not mistaken, it is precisely the feature of self-duality that distinguishes our approach fundamentally
from the lattice theoretic Birkhoff--von Neumann approach 
\cite{V}: both are rooted in {\em projective geometry}, but self-duality is  uninteresting in classical 
approaches, which deal with projective geometries over {\em fields}. The lattice structure of a projective line
over a field is trivial,  hence uninteristing (subspaces are just individual points); over {\em rings}, this changes drastically.


\ssk
Seen from a different angle, working with geometries over rings, as opposed to those over fields, also allows to
integrate aspects of ``fuzzy'', or ``intuitionistic logic'', into our theory, without having to use abstract tools like topos theory
(cf.\ \cite{L}, Chapter 12): namely, in projective geometries over fields, there is just one ``incidence relation'' --
a point belongs to a line, or not: {\it tertium non datur.} 
In geometries over rings, there are a lot of shades of gray, between white
(the point has nothing in common with a line), and black (the point is totally included in the line).\footnote{ A  
one-dimensional submodule over a ring may intersect another subspace non-trivially, without being totally 
included in it.} 
Maybe this viewpoint could add a new facet to the topic of ``quantum logic''.

\subsection{Infinities; completeness} 
The completion of a linear space, such as the linear space of quantum theory, by ``points at infinity'' provides
a convenient and geometric language to speak about ``infinities''. 
Remarkably, it allows to give some sense to ``infinities'' that seem untractable without the geometric framework
(cf.\ Section 2 of \cite{BeKi}).
The bigger the ``set at infinity'' is, the more
it carries structure reflecting complicated analytic or arithmetic structures -- for instance, the ``most difficult base ring''
for our theory is $\K=\Z$ (it has few invertible elements, so we have to add a lot of points at infinity). 
Could this be a good piece of language for speaking about
 \href{https://en.wikipedia.org/wiki/Quantum_field_theory#The_problem_of_infinities}{``problems of infinities''} 
arising in physics?

\ssk
Paradoxically, while insisting on ``geometric completeness'', we relax demands on {\em analytic completeness}:
we prefer to use the more general $P^*$-algebras rather than $C^*$-algebras; they need not be complete in the
anaytic or metric sense. 
For instance 
I try to avoid using Banach space norms altogether (their geometric meaning in the non-linear context is unclear to me). 
And although it is analytically very convenient, on physical
grounds it seems hard to justify that all Cauchy sequences must converge. Statements of the kind
``{\em if} something converges, {\em then}...'' should suffice to cast the logical structure. The uncountable set of all possible 
limits of all possible Cauchy sequences forms another ``infinity'' in the geometric sense: both notions of completeness
and of infinity have non-trivial relations with each other.

\subsection{Completion of commutative and of non-commutative geometry}
According to the basic pattern, {\em commutative} algebras $\bA$ correspond to {\em classical systems}.
This remains true on the level of the ``completed'' theory:
commutative $C^*$-algebras are function algebras, and when $\bA$ is a function algebra, our formalism
of ``complete quantum theory'' corresponds exactly to what has been proposed in Section \ref{sec:class-revisted}:
the projective line over $F(M,\R)$ really is the space $F(M,\R\PP^1)$, and hence ``in the classical limit'', we shall get
back a classical system.\footnote{ If we take algebras of {\em continuous}, or {\em smooth}, 
functions, then some non-trivial analysis is
needed to describe the precise relationship between the projective line over this algebra, and the space of 
all (continuous or smooth) functions with values in $\R\PP^1$.}
In other words, Section \ref{sec:class-revisted} describes the ``completion of commutative geometry''.

\ssk
In the same way, the general ``complete quantum theory'' can be seen as ``completion of Non-Commutative Geometry'',
where 
\href{https://en.wikipedia.org/wiki/Noncommutative_geometry}{\em Non-Commutative Geometry (NCG)} here 
is understood in
its technical sense defined by A.\ Connes. 
Although  methods and aims of NCG appear to be quite different from what is proposed here, I see
no principal obstruction for asking about transferring certain of its methods and results to the ``completed'' setting.
After all, motivation of NCG by physics is often emphasized, so it might turn out that NCG and ``complete quantum
theory''  are complementary, approaching the same reality from different sides.

\subsection{Composed systems}
If $\bA \PP^1$ corresponds to one system and $\bB\PP^1$ to another, then the composed system can be described by
$(\bA \otimes \bB) \PP^1$ -- composition of systems corresponds to {\em tensor product of algebras}.
This idea works well for {\em Jordan-Lie algebras}, and it even distinguishes them among general Jordan algebras, for 
which a tensor product of algebras is missing. 
In fact, this observation was the historical origin for developing the concept of Jordan-Lie algebra in \cite{GP76}, 
going back to ideas on
``composition classes'' by
Niels Bohr -- see \cite{Be08b} for references and some more remarks. 
This, again, motivates to develop a theory describing the geometry corresponding to 
Jordan-Lie algebras (\cite{Bexy}).

\appendix 

\section{$P^*$-algebras}\label{app:P*}

One cannot do mathematics without using formulas. We have avoided them as much as possible in the main text,
but  in the appendices we give precise definitions and formulas for some of the objects mentioned in the main text.
First of all, some definitions related to algebra. 
Algebraists have the habit to work with algebras defined {\em over a general commutative field or ring $\K$}.
We will do the same here; for physicists this may be motivated by the fact that in view of understanding {\em discrete models}
(cf.\ subsection \ref{ssec:contvsdiscrete}) it may interesting to have formalisms that are valid beyond the ``usual choice''
$\K=\R$. 

\begin{definition}
A {\em (binary) algebra} (over a commutative field or ring $\K$) is a linear space over $\K$, together with a bilinear product
map $\bA \times \bA \to \bA$.
If the product is associative, we call $\bA$ an {\em associative algebra}, and often write the product as $a \cdot b$,
or by simple
juxtaposition $ab$.
If the product is skew-symmetric and satisfies the Jaobi-identity, $\bA$ is called a {\em Lie algebra}, and the product
is often denoted by $[a,b]$.
If the product, written $\bullet$, is commutative and satisfies the {\em Jordan identity},
$$
\forall a,b \in \bA : \qquad a \bullet (b \bullet a^2) = (a \bullet b) \bullet a^2, \mbox{ where } a^2 = a \bullet a ,
$$
 then $\bA$ is called a {\em Jordan algebra}.
 \end{definition}

\nin
Every associative algebra gives rise to a family of associative, 
 Lie, and of Jordan algebras, sometimes called {\em homotopes} of each other:
 
 \begin{lemma}
Let $\bA$ be an associative algebra, and  fix
 $u \in \bA$. Then 
\begin{equation}\label{eqn:JL1}
a \cdot_u b = aub, \qquad
[a,b]_u := aub - bua, \qquad
a \bullet_u b := \frac{1}{2} (aub + bua).
\end{equation}
are associative, Lie, respectively Jordan algebra products on $\bA$. 
\end{lemma} 

\nin When $u=1$ is a neutral element, we get the ``usual'' products $ab$,
$[a,b]$ and $a\bullet b$.

\begin{definition}
An associative algebra is called {\em unital} if it has a unit element $1$. Then an element $a \in \bA$ is called {\em invertible}
if there is $b\in \bA$ with $ab =1=ba$. The set $\bA^\times$ of invertible elements then is a group.
\end{definition}

\nin
It is easy to show that $a$ is invertible, if and only if, both the left and right multiplication operators
$L_a (x) = ax$ and $R_a(x)=xa$ are invertible, iff the operator $Q_a(x)=L_a \circ R_a(x) = axa$ is invertible. 
This may serve to define invertible elements even in non-unital associative algebras (see below, def.\
\ref{def:invertibleelement}).

\begin{definition}[$*$-algebra]\label{def:staralgebra}
A  \href{https://en.wikipedia.org/wiki/*-algebra}{\em $*$-algebra} is an associative complex algebra $\bA$ together with an {\em involution} 
$\bA \to \bA$, $a \mapsto a^*$ (that is, a complex anti-linear map such that $(a^*)^*=a$,  $(ab)^* = b^* a^*$, $1^*=1$).
An element $a \in \bA$ is called {\em Hermitian} if $a^*=a$, and {\em skew-Hermitian} if $a^* = -a$.
The sets of {\em (skew) Hermitian elements} are denoted by 
\begin{equation*}
\Herm(\bA)=\{ a \in \bA \mid a^* = a \} ,\qquad
\Aherm(\bA)=\{ a \in \bA \mid a^* = a \} .
\end{equation*}
The {\em unitary group of a (unital) $*$-algebra} is the subgroup of $\bA^\times$ given by
$$
\UU(\bA,*) := \{ a \in \bA \mid a a^* = 1 = a^* a \} = \{ a \in \bA^\times \mid \, a^{-1} = a^* \} .
$$
\end{definition}

\nin
By decomposing $a=\frac{a+a^*}{2}+\frac{a-a^* }{2}$, we see that
$\bA = \Herm(\bA) \oplus \Aherm(\bA)$,
and since $*$ is antilinear, we get $\Aherm(\bA) = i \, \Herm(\bA)$, whence
\begin{equation}
\bA =\Herm(\bA) \oplus i \, \Herm(\bA) .
\end{equation}

\begin{lemma}
Assume $\bA$ is a $*$-algebra, and
consider the products given by (\ref{eqn:JL1}).
If $u^* =u$, then $\Herm(\bA)$ is a Jordan algebra and $\Aherm(\bA)$ a Lie algebra, and if $u^* = -u$, then $\Herm(\bA)$ is a Lie algebra and $\Aherm(\bA)$ is a Jordan algebra.
In particular, $\Herm(\bA)$ is a Jordan algebra for the product $\bullet$ and a Lie algebra for the product
$[--]_{\hbar i}$, for any constant $\hbar \in \R$.
\end{lemma}

\begin{proof}
It suffices to check that the spaces are stable under the products in question, and this follows directly from
$(aub)^* = b^* u^* a^*$.
\end{proof}

\nin
The Jordan and Lie products on $\Herm(\bA)$ satisfy certain natural compatibility conditions: they define a
{\em Jordan-Lie algebra}. This structure is important for the theory of {\em time evolution}, and will be 
investigated in more detail in Part II \cite{Bexy}.

\begin{definition}[$P^*$-algebra]\label{def:P*-algebra}
 A {\em $P^*$-algebra} is a {\em positive $*$-algebra}, that is, a $*$-algebra such that $\Herm(\bA)$
carries a structure of
\href{https://en.wikipedia.org/wiki/Ordered_vector_space}{\em ordered vector space}
(over the real field, or over some other partially ordered ring),
such that, 
\begin{enumerate}
\item
whenever $b \in \Herm(\bA)$ is positive (i.e., $b\geq 0$), then 
$\forall a \in \bA :\,  a b a^* \geq 0$, 
\item
for all $a \in \bA$, and all  invertible
$b \in \bA$, the element
$a^* a + b^* b$ is invertible. 
\end{enumerate}
\end{definition}

\nin
If this holds, then $(\Herm(\bA), \leq )$ is an {\em ordered Jordan algebra} (see \cite{Be17b}).
The second condition is a weakening of the well-known condition of being
\href{https://en.wikipedia.org/wiki/Jordan_algebra#Formally_real_Jordan_algebras}{\em formally real}.

\begin{lemma}
Any $C^*$-algebra (where $x\geq 0$ iff $Spec(x) \geq 0$)
 is a $P^*$-algebra.
 \end{lemma}
 
\nin Indeed, in a $C^*$-algebra, (2) is vacuous if the algebra has no invertible elements; else, 
  $Spec(b^*b) \geq \lambda $ for some constant
$\lambda >0$ if $b$ is invertible, hence the same holds for $a^*a + b^* b$, hence $a^*a  + b^* b$ is invertible.
We prefer to work with $P^*$-algebras, since they are more general than $C^*$-algebras, 
and their defining properties have a clear geometric meaning.

%

 
\section{Associative pairs}\label{app:AP}

Square matrices are generalized by rectangular matrices (including the important special cases of row and column vectors). 
In the same way, usual (binary) algebras (associative, or
Jordan) are generalized by   {\em associative pairs}, resp.\ {\em Jordan pairs}. Both concepts are not very well
known among mathematicians. The Jordan pair concept, as introduced by Loos in \cite{Lo75}, is quite technical,
and we will not use it in this text. The concept of associative pair, on the other hand, is very simple (see \cite{Lo75}, or
 Appendix B in \cite{BeKi}):

\begin{definition}
An \emph{associative pair (over a commutative ring $\bK$)} is a pair $(\bA^+,\bA^-)$ of
$\bK$-modules together with two trilinear maps
\[
\langle \cdot,\cdot,\cdot \rangle^\pm :\bA^\pm \times \bA^\mp \times \bA^\pm  \to
\bA^\pm, \quad
(x,y,z) \mapsto \langle xyz \rangle^\pm
\]
satisfying the following {\em para-associative law}:
\[
\langle xy \langle zuv\rangle^\pm \rangle^\pm =
\langle\langle xyz\rangle^\pm uv\rangle^\pm=\langle x\langle uzy\rangle^\mp v\rangle^\pm.
\]
It is called {\em commutative} if always $\langle x y z \rangle^\pm = \langle z y x \rangle^\pm$.
Fixing the middle element $a \in \bA^\pm$, we get a binary associative product on $\bA^\mp$, denoted by
$\bA_a$ and called the
{\em $a$-homotope}:
\begin{equation}
xz:=
x \cdot_a z := \langle x a z \rangle^\pm .
\end{equation}
\end{definition}

\ssk
\nin{\bf Examples of associative pairs.}

\begin{enumerate}[label={(}\arabic*{)},leftmargin=*]
\item
Every associative algebra $\bA$ gives rise to an associative pair
$\bA^+ = \bA^- = \bA$ via $\langle xyz\rangle^+ = xyz$, $\langle xyz\rangle^- = zyx$.
\item
For  $\bK$-modules $E$ and $F$, let
$\bA^+ = \Hom(E,F)$, $\bA^- = \Hom(F,E)$, and
\[
\langle XYZ\rangle^+ = X \circ Y \circ Z \qquad
\langle XYZ\rangle^- = Z \circ Y \circ X.
\]
Taking $F=\K$, a linear space $E$ and its dual $E'$ form an associative pair.
\item
$('\R,\R')$ is an associative pair, and so is 
$(F(M,'\R),F(M,\R'))$ (cf.\ section \ref{sec:class-revisted}).
\item
Let $\hat{\bA}$ be an associative algebra with unit $1$ and idempotent $e$
(that is, $e^2=e$) 
and $f:=1-e$ its ``opposite idempotent''.
Let
\[
\hat{\bA} = f \hat{\bA} f \oplus f \hat{\bA} e \oplus e \hat{\bA} e
\oplus e \hat{\bA} f =
\bA_{00} \oplus \bA_{01} \oplus \bA_{11} \oplus \bA_{10}
\]
with
$\bA_{ij} = \setof{x \in \hat{\bA}}{ex=ix, xe=jx}$
the associated {\em eigenspace (Peirce) decomposition}. Then
\[
(\bA^+,\bA^-) := (\bA_{01},\bA_{10}), \quad
\langle xyz\rangle^+ := xyz, \quad \langle xyz\rangle^- := zyx
\]
is an associative pair.
\end{enumerate}

\nin
It is easy to show that every associative pair arises from
an associative algebra $\hat{\bA}$ with idempotent $e$ in the way
just described (\cite{Lo75}, Notes to Chapter II).

\begin{definition}[invertible elements] \label{def:invertibleelement}
We call an element $x \in \bA^\pm$ \emph{invertible} if
\[
Q_x:\bA^\mp \to \bA^\pm, \quad y \mapsto \langle xyx\rangle
\]
is an invertible operator.
\end{definition}

\nin As shown in \cite{Lo75}, associative pairs with invertible
elements correspond to unital associative algebras:
namely, $x$ is invertible if and only if the homotope algebra $\bA_x$ has
a unit (which is then $x\inv:=Q_x\inv x$).

\begin{definition}[idempotent]\label{def:idempotent}
An {\em idempotent} in an associative pair is a pair $(e^+,e^-)\in A^+ \times A^-$ such that
\[
\langle e^+,e^-,e^+ \rangle = e^+ \qquad \mbox{ and } \qquad
\langle e^-,e^+,e^- \rangle = e^- \, .
\]
\end{definition}

\nin
Idempotents are a  tool to start to  ``glue together'' the two spaces $\bA^+$ and $\bA^-$.

\section{Projective spaces, projective lines}\label{app:P}

In this appendix we describe the construction of basic ``geometric  spaces'' by using rings and algebras.
For sake of generality, in this appendix $R$ 
is a {\em (possibly non-commutative) ring with unit $1$}
(we reserve the letter $\K$ to {\em commutative} rings), and $W$ a {\em right module} over $R$.
For a first reading, think of $R=\R$ or $\mathbb C$, and $W$ a vector space, say $W = \bC^n$, or a Hilbert space;
but for a second reading, it will be important to allow for $R$ a {\em non-commutative} ring:
namely, the role of $R$ may be taken by some 
$*$-algebra $\bA$.  For a general ring $R$,
recall that a 
\href{https://en.wikipedia.org/wiki/Module_%28mathematics%29}{module} 
over $\K$ is defined like a vector space, except that there is no commutativity of scalars, and therefore we
agree to write scalars always on the right of vectors.

\subsection{Grassmannians}\label{ssec:P1}
Let $W$ be a right $R$-module, together with a direct sum decomposition $W = A \oplus Z$.
We define the {\em Grassmannian of type $A$ and co-type $Z$} to be the set
of all submodules $E$ that are isomorphic to $A$ and admit a complement $E'$ isomorphic to $Z$:
\begin{equation}
\Gras_A^Z(W) := \{ E \subset W \mbox{ submodule}  \mid \, E \cong A, \exists E' \cong Z : \, W = E \oplus E' \} .
\end{equation}
The pair of Grassmannians
\begin{equation}\label{eqn:P1}
(\cX, \cX') = (\Gras_A^Z(W),\Gras_Z^A(W))
\end{equation}
is said to be {\em in duality}.
For instance, the pair $(\Gras_p(\K^{p+q}),\Gras_q(\K^{p+q}))$, where for a field $\K$,
$\Gras_k(\K^n)$ is the {\em Grassmannian of $k$-dimensional subspaces}  of $\K^n$, 
is such a pair.
The {\em general linear group} $\Gl(W)$ acts on $\cX$ and on $\cX'$. This action is transitive: by definition, there are
linear isomorphisms $g_1:A \to E$, $g_2:Z \to E'$, whence $g:=g_1 \oplus g_2 \in \Gl(W)$ sends $A$ to $E$ (and $Z$ to $E'$).
If $R$ is commutative, then the scalars
act trivially on $\cX$ (but else not).

\subsection{Projective spaces, projective lines, self-duality} 
We say that $A$ is a {\em line} if $A$ is isomorphic to the base ring
($A \cong R$), and then call 
$Z$  a {\em hyperplane}  if $W = A \oplus Z$. The {\em projective space of $W$} is the space of all lines admitting a hyperplane
 complement, and its {\em dual projective space} is the space of all hyperplanes: we write
 \begin{equation}\label{eqn:P2}
(\cX,\cX') = (\Gras_R^Z(W),\Gras_Z^R(W)) = (\PP (W), \PP(W)') .
\end{equation}
In the special case $W = R \oplus R = R^2$, with $A$ the first and $Z$ the second factor, this defines the
{\em projective line $R \PP^1$ over $R$}, together with its dual projective line $(R \PP^1)'$:
\begin{equation}
(R \PP^1, (R \PP^1)' ) = (\Gras_R^R(R \oplus R), \Gras_R^R(R \oplus R)) .
\end{equation}
As sets, $R \PP^1$ and  $(R \PP^1)'$ agree: the projective line is {\em self-dual}. 
Both copies may be distinguished by taking different {\em base points}: in the first copy, the base point
$0$ is the first factor $R \times 0 = [(1,0)]$, and in the second copy, the base point $\infty$ is the second factor
$0 \times R = [(0,1)]$, where we write
$[(x,y)] = (x,y)R$ for the right module generated by $(x,y)$. 
Given these base points, there is a {\em natural imbedding}
\begin{equation}\label{eqn:imbedding}
R \to R \PP^1, \quad z \mapsto [(z,1)],
\end{equation}
and $\Gl(2,R)=\Gl(R\oplus R)$ acts, just as in the classical case $R=\bC$, on the affine part $R$
by ``fractional linear transformations'':
\begin{equation}
\begin{pmatrix} a & b \\ c & d\end{pmatrix} \Bigl[ \begin{pmatrix} z \\ 1 \end{pmatrix}\Bigr] =
\Bigl[ \begin{pmatrix} az+b \\ cz+d \end{pmatrix} \Bigr] 
=
\Bigl[ \begin{pmatrix} (az+b)(cz+d)^{-1} \\ 1 \end{pmatrix}\Bigr] .
\end{equation}

\subsection{Transversality, projection operators}\label{ssec:P2}
Let $(\cX,\cX')$ be as in (\ref{eqn:P1}). A pair $(x,a) \in \cX \times \cX'$ is called {\em transversal}, and we write $a\top x$, if
$W$ is the direct sum of $a$ and $x$: $W = a \oplus x$. 
We denote the set of all complementary subspaces of $a$ by 
\begin{equation}\label{eqn:P3}
U_a = \{ x \in \cX \mid x \top a \} .
\end{equation}

\begin{theorem}\label{th:P1}
The set $U_a$ carries a natural structure of an affine space over $R$.
\end{theorem}

\nin
This result is classical, and easily proved. For instance, for any scalar $r\in R^\times$, multiplication by $r$ in
the linear space $(U_a,x)$,
\begin{equation}\label{eqn:scalar}
r_{a,x}:\cX \to \cX, \quad y \mapsto r_{a,x}(y) 
\end{equation}
is given by the matrix $\bigl(\begin{smallmatrix}r&0\\0& 1\end{smallmatrix}\bigr)$ with respecto to the direct sum decomposition
$W=x \oplus a$.  
The following is less classical (to my knowledge,  \cite{BeKi} is the first time it appeared):

\begin{theorem}\label{th:P2}
Fix $(a,b) \in \cX' \times \cX'$ and let $U_{ab}:= U_a \cap U_b$ (set of common complements of $a$ and $b$).
If $U_{ab}$ is not empty, fix an arbitrary element $y\in U_{ab}$.
Then $U_{ab}$ carries a natural group structure with neutral element $y$. In case  $a=b$, this is the
additive group law of the vector space $U_a$ with zero vector $y$, and in case $a \top b$, this group is
isomorphic to the general linear group $\GL(a)$.
\end{theorem}

\nin
Let us, following \cite{BeKi}, describe the group law on $U_{ab}$. Its product shall be denoted by
$x \cdot_y z$, or $(xyz)_{ab}$.
The main tool for defining it are the {\em projection operators}: when $a \top x$, denote by
$P^a_x:W \to W$ the linear projector having kernel $a$ and image $x$. 
Then, clearly,
$P_x^a \circ P_z^a = P_x^a$, 
$P_x^a  \circ  P_x^b = P_x^b$, 
and $P^b_z \circ P_z^a = 0$.
From these rules it follows by direct computation that, whenever $x,y,z \in U_{ab}$, the linear operator
\begin{equation}
M_{xabz} :=P_x^a - P^z_b : W \to W
\end{equation}
is invertible, with inverse $M_{zabx}$. Applying it to $y$ gives the group law:
\begin{equation}
x \cdot_y z = (xyz)_{ab} = M_{xabz}(y) = (P_x^a - P^z_b)(y) .
\end{equation}
And $r_{a,x} = r P^a_x + P^x_a$.
Much more can be said about this (cf.\ loc.\ cit.)

\begin{example}
When $\cX = \PP(W)$ is a projective space and $a$ a hyperplane, then $U_a$
is identified with the set of all points of $X$ that do not belong to the hyperplance defined by $a$,
$U_a = \cX \setminus H_a$, and $H_a$ is the ``horizon of $U_a$'' (set of points at infinity of the affine space $U_a$).
In this case, there exists a simple lattice-theoretic formula describing 
 the group law of $U_{ab}$.
\end{example}

\begin{example}
If $\cX = R\PP^1= \Gras_R^R(R\oplus R)$ is a generalized projective line, then the image of the imbedding
(\ref{eqn:imbedding}) is $U_\infty$, isomorphic to $(R,+)$ as a group. Likewise, $U_0 \cong (R,+)$ as a group,
and $U_{0\infty}= U_\infty \cap U_0 = R^\times$ is the multiplicative group of the ring $R$ (which is possibly
non-commutative, according to our assumption).
\end{example}

\subsection{The Hermitian projective line}\label{ssec:Hpl}
Now consider the case where the ring $R$ is a $*$-algebra $\bA$, say over $\K=\bC$.
Since $\bA \PP^1$ generalizes the Riemann sphere, we also write
$\cS:=\bA \PP^1 = \Gras_\bA^\bA(\bA^2)$.
The involution $*:\bA \to \bA$ induces an involution $\tau:\cS \to \cS$.  
This involution is given by taking the orthocomplement $\tau(x) :=x^{\perp,\omega}$
of the submodule $x$ with respect to the ``Poisson form'', that is, the skew-Hermitian sesquilinear form
$\omega:\bA^2 \times \bA^2 \to \bA$ given by
\begin{equation}\label{eqn:omega}
\omega((u_1,u_2),(v_1,v_2))  = \begin{pmatrix} u_1^* & u_2^* \end{pmatrix}  J
\begin{pmatrix} v_2\\ v_1 \end{pmatrix} = 
u_1^* \,  v_2 - u_2^* \,  v_1 .
\end{equation}
where $J$ is the matrix
\begin{equation}\label{eqn:J}
J := \begin{pmatrix} 0&1\\ -1 & 0 \end{pmatrix} .
\end{equation}
Indeed, the orthocomplement of $[(1,a)]$ is $[(1,a^*)]$, because for all $x,y \in \bA$,
$$
\omega ((x,ax),(y,a^*y)) = x^* (a^* y) - (ax)^* y = x^* a^* y - x^* a^* y = 0,
$$
i.e., $\tau(a)=a^*$ on $\bA \subset \cS$. (For general elements of $\cS$, there is no ``closed formula'' describing $\tau$, unless
$\bA$ is commutative.)
The fixed point set of this involution is called the {\em Hermitian projective line over $(\bA,*)$}; it
generalizes the ``equator of the Riemann sphere'', and it
 is the set of {\em $\omega$-Lagrangian subspaces} which we denote by 
 \begin{equation}
\cR:=\Herm(\bA)\PP^1 = \Lag_\omega(\bA^2) = \{ x \in \cS \mid x^{\perp,\omega}=x \} .
\end{equation}
The self-dual geometry $(\cR,\cR)$ is the total space for ``completed quantum theory''. 
Its automorphism group is $\Aut(\cR)=\Aut(\cS,\tau) = \PP \Aut(\omega)=
\Aut(\omega)/Z$, the ``symplectic group'' of $\omega$, modulo its center.
It acts transitively on $\cR$.
In matrix form, the group $\Aut(\omega)$, and its Lie algebra $\Der(\omega)$ can be described by
$2\times 2$-matrices
\begin{align}\label{eqn:autR}
\Aut(\omega) & = \{ g \in \Gl(2,\bA) \mid { }^t g^*  J  g =  J  \},
\\
\Der(\omega)& = \{ M \in M(2,2;\bA) \mid { } ^t M^* J =
- J M \}
\\
& = \Bigl\{ \begin{pmatrix} a & b \\ c & -a^* \end{pmatrix} \mid a,b,c \in \bA , b^* =b, c^* =c \Bigr\} .
\end{align}
As a real vector space, 
$\Der(\omega) = \Herm(\bA)\oplus \bA \oplus \Herm(\bA)$, and
when $\bA = \bC$, we see that $\Der(\omega) = \mathfrak{sl}_2(\R) \oplus i \R$, 
where $i \R$ is the center.
Finally, in Subsection \ref{ssec:strong} we consider also the orthocomplementation map $\alpha$ with respect to the 
positive ``scalar product'' on $\bA^2$, whose automorphism group is the unitary group $\UU(2,\bA)$.

\subsubsection{Toy model} 
In the following table (second column) we give formulae describing the special case of the algebra
$\bA = M(n,n;\bC)$ with involution $A^* = \overline A^t$ (conjugate transpose); this case is a good
finite dimensional  ``toy model'' for quantum mechanics
($n=1$ is ``classical''; $n=2$ is the ``qubit''). 

\ssk
\begin{center}
\begin{tabular}{|l|c|c |}
  \hline
   & Hilbert space setting:   & $*$-algebra setting:   \\
   & $\cH = \bC^n$, $\langle u,v\rangle = \sum_i \overline u_i  v_i$ & $(\bA,*)$
   \\
  \hline
complex associative algebra  &  $\End(\cH) = M(n,n;\bC)$   &   $\bA$  \\
 \hline
involution $*$  & $a^* = \overline a^t$  &  $*$
  \\
 \hline
 observables (Jordan algebra)  & $\Herm(\cH)=\Herm(n,\bC)$  &  $\Herm(\bA)$  \\
  \hline
 projective line $\bA\PP^1$ & $\Gras_\cH(\cH\oplus \cH)=\Gras_n(\bC^{2n})$ & $ \Gras_\bA^\bA(\bA \oplus \bA)$ \\
 \hline
 automorphism group $\Aut(\bA\PP^1)$ & $\PP \Gl(2n,\bC)$ & $\Gl(2,\bA)/\mbox{Center}$
 \\
 \hline
finite part $\bA \subset \bA \PP^1$ & $\{ {\rm Graph}_a 
 \mid a \in \End(\cH)\}$ & $\{ [(1,a)] \mid a \in A\}$, \\  
& $\Graph_a=\{ (x,ax) \mid x \in \cH\}$ & $[(u,v)] =(u,v)\, \bA $
\\
\hline
base point $0$ & $\Graph_0 = \cH \times 0$ & $[(1,0)]$
\\
\hline
base point $\infty$ & $0 \times \cH$ & $[(0,1)]$
\\
\hline
unit $1$ & $\Graph_{\id}={\rm dia}(\cH \times \cH)$ & $[(1,1)]$
\\
\hline
Poisson form $\omega((u_1,u_2),(v_1,v_2))$ & 
$\ldots= \langle u_1,v_2\rangle - \langle u_2,v_1 \rangle$ & 
$ \ldots = u_1^* \,  v_2 - u_2^* \,  v_1$
\\
\hline
involution of $\bA\PP^1$ & $\omega$-orthocomplement
& $x \mapsto x^{\perp,\omega}$
\\ 
\hline
Hermitian  projective  line $\cR$ &
$\Lag_\omega(\cH\oplus \cH)$ & $\Lag_\omega(\bA^2)$
\\
\hline
$G=\Aut(\cR)$ & $\Aut(\omega) \cong \UU(n,n)/\mbox{Center}$ & $\Aut(\omega)/\mbox{Center}$
\\
\hline
\end{tabular}
\end{center}

\msk
\nin
The formulae in the first column are obtained from the general formulae of the second column
by considering a $(2n)\times(2n)$-matrix as a
$2 \times 2$-matrix with entries in the algebra $\bA = M(n,n;\bC)$.
Note,  however, that in the general case (second column) $\bA$ may be an infinite dimensional
algebra over $\bC$, and all the preceding groups are infinite dimensional Lie groups (and they are much ``bigger'' than
those usually considered in quantum mechanics: they contain many ``hidden variables'', that is, degrees of
freedom that are not activated in usual quantum mechanics). 
With some care, fomulae from the ``toy model'' generalize to the infinite dimensional Hilbert space setting of
quantum mechanics (to get conceptual formulae, it may be useful to replace $\bC^{2n}$ by 
$\cH \oplus \cH'$, the direct sum of a Hilbert space and its dual space).


\section{Cross-ratio}\label{ssec:CR}

The cross-ratio is the most important invariant of a projective line. We approach it in two steps:
first of all, we recall the classical definition for the projective line over a {\em commutative} field or ring,
and second, we discuss generalizations to the {\em non-commutative case}.

\subsection{The classical cross-ratio}
Retain notation from the preceding appendix, and
assume $R$ is a commutative field henceforth denoted by $\K$. 
 We denote by $W':=\Hom(W,\K)$ the algebraic dual space of $W$.
Then $(X,X')=(\PP(W),\PP(W'))$ are projective spaces in duality (a hyperplane $H$ in $W$ corresponds to
$\ker(A)$ where $A:W \to \K$ is determined up to a scalar).
We define the {\em cross-ratio} of
 a quadruple
$(x,y,a,b)=([\xi],[\eta],[A],[B]) \in X^2 \times (X')^2$ 
(so vectors $\xi,\eta$ and linear forms $A,B$ are defined up to a scalar) by
\begin{equation}\label{eqn:CRR} 
\CR (x,y;a,b) := \frac{A(\xi)}{A(\eta)} : \frac{B(\xi)}{B(\eta)} =
 \frac{A(\xi) \cdot B(\eta)}{A(\eta) \cdot B(\xi)} \, .
\end{equation}
Note that this is well-defined (independent of scaling of $\xi,\eta,A,B$), and it clearly is an {\em invariant} under
the natural action of the general linear group $\Gl(W)$, acting as usual on $W$ and on $W'$. 
In other words, this definition defines a {\em natural invariant of projective spaces and their duals}. 

\ssk
Now assume that $W$ is two-dimensional, say  $W =\K^2$.
The special feature of this case is the existence of a {\em canonical symplectic form} $\omega$:
it is given by the same formula as (\ref{eqn:omega}), with involution the identity map:
\begin{equation}
\omega : \K^2 \times \K^2 \to \K, \quad ((x_1,x_2),(a_1,a_2)) \mapsto x_1 a_2 - x_2 a_1  .
\end{equation}
This form is $\K$-bilinear (commutativity of $\K$ is crucial here!),
and up to a factor it is invariant under the whole linear group $\Gl(2,\K)$.
We may use it in order to identify  $X$ and $X'$, 
so that our invariant (\ref{eqn:CRR}) is turned into a function defined on $X^4$ and given by 
\begin{equation}
\CR(x,y;a;b) = \frac{\omega(A,\xi)}{\omega(A,\eta)}:\frac{\omega(B,\xi)}{\omega(B,\eta)}=
\frac{ (a_1 x_2 - a_2 x_1) (b_1 y_2 - b_2 y_1)}{ (a_1 y_2 - a_2 y_1)(b_1 x_2 - b_2 x_1)} .
\end{equation} 
Letting in this formula $x_2 = 1 = a_2 = y_2 = b_2$, we get the value
$\frac{(a_1 - x_1)(b_1 - y_1)}{(a_1 - y_1)(b_1 - x_1)}$, which corresponds to the ``usual'' definition of the
cross-ratio $\CR(x_1,y_1;a_1,b_1)$, as given by formula (\ref{eqn:CR}).
Hence both definitions are in keeping. The one given here has several advantages: it features {\em duality}, and
it shows where {\em self-duality} and {\em commutativity of $\K$} enter into the definition.

\subsection{Operator valued cross-ratio}\label{ssec:P3}\label{ssec:GCR}
We wish to define analogs of the cross-ratio in
 the general setting of \ref{ssec:P1}. In fact, it is always possible to define
 an ``operator valued cross-ratio'', but  it may be problematic to extract from
it a  {\em scalar valued} invariant -- to do this, one needs things like
{\em determinants}, or {\em traces}, and these may not exist in infinite dimension. 
Let us explain this: assume  $a \top x$. Then the linear spaces
$(U_x,U_a)$, with origin $(a,x)$, are in duality with each other, in the following sense:
with respect to the decomposition $W = a \oplus x = a \times x$,
every element $b \in U_x$  can be written as
 the {\em  graph} of a unique linear map
$\beta : a \to x$, and every element $y \in U_a$ as graph of a linear map
$\eta: x \to a$.
Thus, as pair of linear spaces,
\begin{equation}
(U_x,U_a) = (\Hom_R (a,x), \Hom_R(x,a) ) .
\end{equation}
Note that this pair is again an associative pair. This observation leads to define
 two   ``canonical kernel functions''
\begin{align}\label{eqn:GCR1}
K_{x,a} : U_x \times U_a \to \End(x), &  \quad  (\beta,\eta) \mapsto \beta \circ \eta , \\
K_{a,x} : U_a \times U_x \to \End(a),  & \quad  (\eta,\beta) \mapsto \eta \circ \beta .
\end{align}
Now we define the {\em generalized (operator valued) cross-ratio}   by
\begin{align}\label{eqn:GCR2}
\CR(y,b; x,a) := K_{x,a}(b,y) \in \End(x)
\end{align} 
(so $\CR(y,b;0,\infty) = by$).
The construction is natural, that is, invariant under the
symmetry group $\Gl(W)$: 
\begin{equation}
\forall g \in \Gl(W): \qquad K_{gx,ga}(gb,gy)=g K_{x,a} (b,y) g^{-1} .
\end{equation}
The operators now  live in a space depending on $x$ (or on $a$). In technical terms, they define an {\em invariant
section of a vector bundle}. 
Moreover, at least when $a$ and $b$ are transversal, the operator valued cross ratio is closely related to
the {\em left, right and middle multiplication operators} defined in \cite{BeKi}.


\subsection{Scalar valued cross-ratio (expectation value)}
To extract a well defined scalar from the operator valued cross-ratio, we have to compose the $\End(x)$-valued cross-ratio
with a function
$\End(x) \to \K$ that is {\em conjugation invariant}, such as  {\em determinant} or {\em  trace functions}
(see Appendix \ref{app:TDI} for more on traces),
\begin{equation}\label{eqn:CI}
\det(g A g^{-1}) = \det(g), \qquad \tr(g A g^{-1} ) = \tr (A) , 
\end{equation}
giving two candidates to define a scalar valued cross-ratio:
\begin{equation}\label{eqn:GCR3}
(x,y;a,b) \mapsto \det (K_{x,a}(b,y)), \qquad
(x,y;a,b) \mapsto \tr(K_{x,a}(b,y)) .
\end{equation}
More generally, $\det$ and $\tr$  could be replaced here by any map
$\chi:\End(x) \to \K$ that satisfies (\ref{eqn:CI}).
In the same way, when $\bA$ is a non-commutative algebra over $\R$, 
there is an operator valued cross-ratio on
 $\bA\PP^1$; but to extract from it an $\R$-scalar valued one, we again need some conjugation invariant
map $\chi:\bA \to \R$. 
Defining such maps is, in infinite dimension, closely related to {\em integration theory}, see the following
appendix.

\section{Pairings, densities, and traces}\label{app:TDI}

Recall from subsection \ref{sec:class-revisted} that $'\R$ and $\R'$ are just two copies of $\R$, without  a fixed base,
and $\R'$ is the dual space of $'\R$.
If $M$ is, say, a topological space, or a general measurable space, we denote in this appendix by 
$F(M, \, '\R)$ and $F(M,\R')$ the set of all measurable functions on $M$. 

\begin{definition}
A {\em pairing on $(F(M, { } '\R), F(M,\R'))$} is a map 
\begin{equation*}
\Pi: F(M,\, '\R) \times F(M,\R') \to \R \cup \{ \infty\} = \R \PP^1, \quad (f,g) \mapsto \Pi(f,g)
\end{equation*}
having the following properties: 
\begin{enumerate}
\item 
it is $\R$-bilinear (in the sense of 
\href{https://en.wikipedia.org/wiki/Riemann_sphere#Arithmetic_operations}{extended arithmetic operations} including rules like
 $x+\infty = \infty$, $\lambda \infty = \infty$, for $x,\lambda \in \R$),
\item
for all $f,g,h \in F(M,\R)$, we have $\Pi (fh,g)=\Pi(f,hg)$, 
\item
it is positive: if $f\geq 0$ and $g\geq 0$, then $\Pi(f,g) \geq 0$,
\item
it is $\sigma$-continuous: if $f_n \downarrow 0$ $(n\to \infty)$ (pointwise monotone convergence), then
$\Pi (f_n ,g) \downarrow 0$ ($n\to \infty$) (whenever $\Pi (f_N ,g) \not=\infty$ for at least some $N \in \N$), and likewise in
the second argument.
\end{enumerate}
\end{definition}


\nin
If $\mu$ is a measure on $M$, we write $\mu(f)=\int_M f d\mu$, and then the formula
\begin{equation}\label{eqn:PP1}
\Pi_\mu (f,g) := \mu (fg)
\end{equation}
defines a pairing. 
This pairing is ``invariant'' in the following sense:
let $G$ be the group of measurable bijections of $M$ preserving the collection of sets of measure zero. 
An element $\phi \in G$  need not preserve $\mu$, but
$\phi_* \mu (h) := \mu (h \circ \phi)$ defines another measure $\phi_* \mu$, which 
is absolutely continuous with respect to
$\mu$.  Thus, by the Radon-Nikodym theorem, there is a function $\phi'$ with $\phi_* \mu = \phi' \mu$, i.e.,
for all $h$,
\begin{equation}\label{eqn:RN}
\mu(h \circ \phi) = \mu( \phi' \cdot h) . 
\end{equation}
(to be precise, $\phi'$ is defined up to sets of measure zero). Applying  (\ref{eqn:RN}) twice, we get the ``chain rule'':
 for all $\phi,\psi \in G$, we have 
$(\phi \circ \psi)' = \phi' \cdot (\psi' \circ \phi^{-1})$. 
Now we let act $G$ in the ``usual'' way by $\phi.f = f\circ \phi^{-1}$
on ``usual'' functions $F(M,\,'\R)$, and via
\begin{equation}
\phi. h :=  \phi' \cdot (h \circ \phi^{-1})
\end{equation}
on $F(M,\R')$. By the ``chain rule'', this is indeed an action. 
When equipped with this action, we call the space $F(M,\R')$ the
{\em space of $\mu$-densities}. 
Now, using (\ref{eqn:RN}),
\begin{align*}
\Pi_\mu (\phi.f,\phi.h) & = \mu (f\circ \phi^{-1} \cdot \phi' \cdot g \circ \phi^{-1} ) \\
& = \mu ( ((f\circ \phi^{-1})\cdot (h \circ \phi^{-1})) \circ \phi ) \\
& = \mu ( f h) =\Pi_\mu (f,h) .
\end{align*}
This proves:

\begin{proposition}
With notation as above,  the pairing (\ref{eqn:PP1}) is  invariant under the group $G$: for all $\phi\in G$, we have
$\Pi_\mu (\phi.f,\phi.h)  = \Pi_\mu (f,h)$.
\end{proposition}

\nin
The following examples illustrate that, under natural assumptions, the pairing
can be considered as ``canonical'' --  it does not really depend on $\mu$, but only on the
class of measures having the same sets of measure zero as $\mu$: 
\begin{enumerate}
\item
If $M$ is a differentiable manifold, then we define integration with respect to volume forms, as usual.
Then densities in our sense coincide with 
\href{https://en.wikipedia.org/wiki/Density_on_a_manifold}{those in the sense of differential geometry},
and our trace is the  pairing between densities and (say) continuous functions. 
This pairing is invariant under the group of all diffeomorphisms, which is a subgroup of $G$.
\item
Generalizing the preceding item,  whenever we have a partition of unity
subordinate to an atlas of $M$, then pairings defined with respect to  chart domains can be glued together to give
a pairing on $M$.
\item
If $M$ is a finite set, then $\tr (f, g) :=\sum_{p\in M} f(p) g(p)$ defines a pairing. It is invariant under the group of all bijections
of $M$.
\end{enumerate}
Once we have a  theory of pairings on the ``classical pair''
$(F(M, \, '\R),F(M,\R'))$, one would like to develop such a theory on more general associative pairs $(\bA^+,\bA^-)$.
In case $\bA^+ = \bA^- = \bA$ is a $C^*$-algebra, this should more or less correspond to spectral theory, and hence
developing such a theory is a big task clearly exceeding the scope of the present work. 
We shall just state the following definitions:

\begin{definition}\label{def:trace}
Let $(\bA,*)$ be a $P^*$-algebra (def.\ \ref{def:P*-algebra}).
A {\em trace} on $\bA$ is  a linear map (where ``linear'' is understood in the generalized sense, as above)
$$
\tr :  \bA \to \bC \PP^1
$$
that is symmetric: $\tr(ab)=\tr(ba)$ and
positive: whenever $b \in \Herm(\bA)$ is positive, then 
$\tr(b) \geq 0$.
It may be normalized by the following condition:
if $a \in \Herm(\bA)$  is an idempotent of rank one, then $\tr(a) = 1$.
\end{definition}

\nin
More geometrically, the trace map could be defined as an ordered morphism of projective lines
$\bA \PP^1 \to \bC \PP^1$, and possibly normalized by the condition that on each ``intrinsic projective
line'' $\bC \PP^1$ contained in $\bA \PP^1$, it should induce the identity mapping. 
Given such a trace function, the binary map
$(a,b)\mapsto \Pi(a,b):= \tr (ab)$ then should define what one might call a 
``pairing on $(\bA,\bA)$'', and
$(x,y;a;b) \mapsto \tr( K_{x,a}(b,y))$ an ``expectation value'' (scalar valued cross-ratio).


\end{document}